\documentclass[11pt]{article}
\usepackage{amsmath,amssymb,amsthm,a4wide}
\usepackage{latexsym}
\usepackage{indentfirst}
\usepackage{placeins}
\usepackage{booktabs}
\usepackage{algorithm}
\usepackage{algorithmic}
\usepackage{multirow}
\usepackage{color,xcolor,colortbl}
\usepackage{thmtools}
\usepackage{authblk}
\usepackage{subfigure}
\usepackage{graphicx,color}
\usepackage{diagbox}

\usepackage{hyperref}

\usepackage{cleveref}
\usepackage{braket}
\usepackage{nicematrix}
\usepackage{extarrows}

\theoremstyle{plain}
\newtheorem{theorem}{Theorem}
\newtheorem{lemma}[theorem]{Lemma}
\newtheorem{prop}[theorem]{Proposition}
\newtheorem{cor}[theorem]{Corollary}

\newtheorem{result}[theorem]{Result}
\newtheorem{definition}[theorem]{Definition}

\theoremstyle{definition}

\theoremstyle{remark}
\newtheorem{remark}{Remark}
\DeclareMathOperator{\polylog}{polylog}
\DeclareMathOperator{\poly}{poly}

\DeclareMathOperator{\Arg}{Arg}

\newcommand{\abs}[1]{%
  \mathchoice
    {\left|#1\right|} 
    {|#1|}           
    {|#1|}           
    {|#1|}           
}

\newcommand{\norm}[1]{%
  \mathchoice
    {\left\lVert#1\right\rVert} 
    {\lVert#1\rVert}           
    {\lVert#1\rVert}           
    {\lVert#1\rVert}           
}

\newcommand{\bbm}{\begin{bmatrix}}
	\newcommand{\ebm}{\end{bmatrix}}
\newcommand{\RR}{\mathbb{R}}
\newcommand{\CC}{\mathbb{C}}
\newcommand{\DD}{\mathbb{D}}

\newcommand{\ZZ}{\mathbb{Z}}

\newcommand{\rd}{\,\mathrm{d}}

\renewcommand{\Re}{\mathrm{Re}}
\renewcommand{\Im}{\mathrm{Im}}

\renewcommand{\Re}{\mathrm{Re}}
\renewcommand{\Im}{\mathrm{Im}}
\newcommand{\I}{\mathrm{i}}

\newcommand{\half}{\frac{1}{2}}
\newcommand{\mo}[1]{\mathcal{O}\left(#1\right)}

\newcommand{\Res}{\text{Res}}

\newcommand{\ba}{\mathbf{a}}

\newcommand{\bk}{\mathbf{k}}

\newcommand{\bw}{\mathbf{w}}
\newcommand{\bx}{\mathbf{x}}

\newcommand{\op}[2]{\mathrm{Op}\left(#1\right)\left(#2\right)}
\newcommand{\sch}{\mathcal{S}(\RR^2)}
\newcommand{\matn}{\mathrm{Mat}_N(\CC)}
\newcommand{\supp}{\mathrm{supp}}

\newcommand{\wh}[2]{\mathrm{Op}_{H}(#1)\left(#2\right)}

\newcommand{\bxi}{\boldsymbol{\xi}}

\usepackage{tikz}
\usetikzlibrary{decorations.pathmorphing}

\title{Quantum Eigenvalue Transformation via Linear Combination of Hamiltonian Simulation: A Weyl Calculus Approach\thanks{This work is partially supported by the U.S. Department of Energy, Office of Science, Accelerated Research in Quantum Computing Centers, Quantum Utility through Advanced Computational Quantum Algorithms, grant no. DE-SC0025572.}}
\author[]{Hongkang Ni\thanks{hongkang@stanford.edu} }
\author[]{Lexing Ying\thanks{lexing@stanford.edu}}
\affil[]{Stanford University, Stanford, CA 94305, United States}

\begin{document}
\date{}
\maketitle

\begin{abstract}
    Linear combination of Hamiltonian simulation (LCHS) provides an efficient method for implementing matrix exponentials $e^{-tA}$ on quantum computers. In this paper, we develop LCHS formulas for computing general matrix functions $f(A)$ when $f$ is analytic on the numerical range of $A$, with $A$ possibly non-normal. The essential technical tool is Weyl calculus, which reduces the construction of LCHS formulas for noncommuting operators to scalar Fourier approximation problems. Our construction yields a quantum eigenvalue transformation algorithm with optimal $\mo{\log\frac{1}{\epsilon}}$ query complexity scaling. Furthermore, our Weyl-calculus-based theory gives rise to an ansatz-free convex optimization framework that directly produces discrete LCHS formulas. This circumvents the inefficiencies of traditional quadrature rules and yields formulas highly optimized for coherent implementation on quantum computers. In addition, both our theory and optimization framework apply to the simulation of time-dependent dissipative ODE $\frac{\rd}{\rd t} \psi(t) = -A(t)\psi(t)$, for which we achieve a $2.1\times$ cost reduction over prior art.
\end{abstract}

\tableofcontents

\section{Introduction}
Quantum computers are expected to solve certain computational problems more efficiently than classical computers by exploiting high-dimensional quantum state spaces and coherent quantum dynamics. Potential applications include scientific computing, quantum simulation, optimization, finance, and machine learning; for instance, see the surveys \cite{dalzell2023quantum, auyeung2024quantum}.

Many quantum algorithms for linear algebra and differential equations can be understood as eigenvalue transformations. If a matrix $A\in\CC^{N\times N}$ is diagonalizable as $A=VDV^{-1}$ and $f$ is defined on the spectrum of $A$, the eigenvalue transformation associated with $f$ is
\begin{equation}
    f(A):=Vf(D)V^{-1}.
\end{equation}
For analytic $f$, this definition agrees with the usual matrix functional calculus and extends to non-diagonalizable matrices by continuity. The corresponding quantum task is to block-encode $f(A)$ in a quantum circuit, or to prepare a normalized state proportional to $f(A)\ket{\psi}$ for an input state $\ket{\psi}$. The well-known QSVT framework \cite{gilyen2019quantum} solves this problem for Hermitian $A$, but it does not apply to general matrices.

For the important special case $f(z)=e^{-tz}$ on the right half-plane $\{\Re(z)\ge 0\}$, the linear combination of Hamiltonian simulation (LCHS) approach is particularly effective. For a matrix $A$, write its Cartesian decomposition as $A=X+\I Y$, where $X$ and $Y$ are Hermitian. Under the dissipativity assumption $X = \frac{A+A^\dagger}{2}\succeq 0$, prior works \cite{an2023linear, an2026quantum, huang2025fourier} prove formulas of the form
\begin{equation}\label{eq: first lchs formula}
    e^{-tA} = (2\pi)^{-1/2}\int_{-\infty}^\infty \hat{h}(u) e^{-\I t(uX+Y)} \rd u
\end{equation}
for suitable kernels $\hat h$. Explicit examples include
\begin{equation}\label{eq: LCHS exp(-tA)}
    \hat{h}(u) = \sqrt{\frac{2}{\pi}}\frac{1}{1+u^2},
    \quad \text{or}\quad
    \hat{h}(u) = \frac{e^{2^\beta-(1+\I u)^\beta}}{\sqrt{2\pi}(1-\I u)},\quad 0<\beta<1.
\end{equation}
These identities are nontrivial generalizations of the corresponding scalar identities, because $X$ and $Y$ generally do not commute. On a quantum computer, the unitary factors $e^{-\I t(uX+Y)}$ can be implemented by well-studied Hamiltonian simulation algorithms \cite{low2017optimal, low2019hamiltonian, gilyen2019quantum}. After truncation and discretization, \eqref{eq: first lchs formula} yields quantum circuits that encode $e^{-tA}$ to additive precision $\epsilon$ using $\mo{t\log^{1+\mo{1}}\frac{1}{\epsilon}}$ queries to block-encodings of $A$ and $A^\dagger$.

Ref.~\cite{low2025optimal} observed that certain truncated kernels $\hat h$ need not satisfy the exact identity \eqref{eq: first lchs formula}, but can still satisfy an approximation of the form
\begin{equation}\label{eq: second lchs formula}
    e^{-tA} \approx (2\pi)^{-1/2}\int_{-R}^R \hat{h}(u) e^{-\I t(uX+Y)} \rd u    
\end{equation}
for some finite $R$. This relaxation leads to quantum algorithms with improved query complexity $\mo{t\log\frac{1}{\epsilon}}$.

From the viewpoint of \cite{low2025optimal, huang2025fourier}, the above LCHS formulas can be interpreted as follows: one first finds a function $h(x)$ that approximates $e^{-x}$ on $[0,\infty)$, and then uses the Fourier transform $\hat h$ as the LCHS kernel. A natural question is whether this perspective can be generalized to other functions $f(z)$.

For $e^{-z}$, the scalar identity $e^{-(x+\I y)}=e^{-x}e^{-\I y}$ separates the decay in $x$ from the phase in $y$, and the phase factor is incorporated into the Hamiltonian simulation term. For a general analytic function, such a separation is unavailable. It is therefore natural to Fourier transform in both the real and imaginary directions, leading to the following question:

\begin{quote} \textit{If $h(x,y)$ approximates $f(x+\I y)$ in some region, will the formula
\begin{equation}\label{eq: lchs formula}
    f(A) \approx (2\pi)^{-1}\int_{u^2+v^2\le R^2} \hat{h}(u,v) e^{-\I(uX+vY)} \rd u \rd v 
\end{equation}
hold in some sense?
} \end{quote}

In this paper, we give a positive answer to this question using \textit{Weyl calculus}. 

Even with this representation, the continuous LCHS formula \eqref{eq: lchs formula} is not yet directly suitable for coherent quantum circuit implementation. In practice, one usually needs a discrete LCHS formula
\begin{equation}\label{eq: disc lchs formula}
    f(A)  \approx \sum_{(u_j,v_j)\in \mathcal{G}\subset D_R} w_{j}\, e^{-\I(u_jX+v_jY)} 
\end{equation}
where $\mathcal{G}$ is a grid contained in the disc $D_R=\{(u,v) \mid u^2+v^2 \le R^2\}$. In previous LCHS works, such discrete formulas are obtained by truncating the continuous formula and then applying quadrature rules. This approach has two limitations. First, explicitly constructed kernels $\hat h$ may have suboptimal constants or scalings, and this inefficiency is inherited by the discretized formula. Second, quadrature rules are most effective for smooth or analytic kernels. If $\hat h$ is not sufficiently smooth, even when $\hat h$ decays rapidly, the number of discretized terms can become large. Based on our Weyl calculus framework, we develop a convex optimization method that directly finds good discrete formulas of the form \eqref{eq: disc lchs formula} under a prescribed resource budget. Interestingly, applying this optimizer to $f(z)=e^{-z}$ suggests a closed-form kernel for \eqref{eq: second lchs formula} that improves the cost relative to prior art \cite{low2025optimal}.

For the quantum implementation of \eqref{eq: disc lchs formula}, the cost is largely determined by two parameters. The first is the truncation radius $R$, which controls the maximum simulation time in the Hamiltonian simulation circuits. The second is the normalization factor $\lambda = \sum_j |w_j|$ (or $\lambda = (2\pi)^{-1}\norm{\hat{h}}_{L^1(\RR^2)}$ in the continuous case \eqref{eq: lchs formula}), which controls the postselection probability in LCU-based implementations. Amplitude amplification \cite{brassard2002quantum} can compensate for this normalization but at the cost of proportionally increasing the simulation time. Therefore, this paper uses
\begin{equation}
    \Lambda = \lambda R
\end{equation}
as the cost function of an LCHS formula.

\subsection{Main results}

\subsubsection{Weyl-calculus framework for general LCHS}
We develop a general LCHS theory based on Weyl calculus, providing a systematic way to construct LCHS formulas for analytic matrix functions. 

\begin{result}[Informal version of \Cref{thm: opf bound}]\label{res: 1}
    Denote $\Omega = \{\bx^\dagger A \bx \mid \bx \in \mathbb{C}^N, \|\bx\|_2 = 1\}$ to be the numerical range of $A$. Let $f$ be analytic on a neighborhood of $\Omega$. If a function $h: \RR^2 \to \CC$ satisfies
    $$f(x+\I y) \approx h(x,y),\qquad \text{for } x+\I y \in \Omega,$$
    then we have the approximate LCHS formula
    $$f(A) \approx (2\pi)^{-1}\int_{\RR^2} \hat{h}(u,v) e^{-\I(uX+vY)} \rd u \rd v,$$
    where $\hat{h}$ is the Fourier transform of $h$.
\end{result}

Moreover, the framework allows $\hat h$ to be a finite measure, such as a Dirac comb, so it directly covers discrete weight distributions.

\subsubsection{Closed-form LCHS constructions for analytic functions}
Using the Weyl calculus framework, we construct LCHS formulas for analytic functions and derive explicit bounds on the cost function $\Lambda$. 

\begin{result}[Alternative version of \Cref{thm: construction 2 for disc}]\label{res: 2}
Assume $\norm{A}\le 1$. Let $\rho\in(0,1)$, and suppose that $f$ is analytic in the interior of disc $\DD_{1+\rho}=\{z\in\CC: |z|\le 1+\rho\}$ and satisfies
    $$\max_{\abs{z}\le 1+\rho} \left(\max\{\abs{f(z)}, \abs{f'(z)},\abs{f''(z)}\}\right)\le S $$
    for some $S$. Then there exist $R = \mo{\frac{1}{\rho}\log\frac{S}{\rho\epsilon}}$ and a smooth function $\hat{h}$ such that
    \begin{equation}
    \norm{f(A) - (2\pi)^{-1}\int_{u^2+v^2\le R^2} \hat{h}(u,v) e^{-\I(uX+vY)} \rd u \rd v}\le \epsilon 
    \end{equation}
    with the cost function
    $$\Lambda = \lambda_{\hat{h}}R = \mo{\frac{S}{\rho}\log\frac{S}{\rho\epsilon}}.$$
\end{result}
We also provide a sharper upper bound in \Cref{thm: construction 1 for disc} by replacing $S$ with the local Wiener norm $\norm{f}_{W(\DD_{1+\rho})}$, defined in \eqref{eq: restricted wiener norm def}. However, this sharper bound is less explicit because the local Wiener norm is usually difficult to compute exactly. Additionally, if the assumption $\norm{A}\le 1$ is replaced by a more general numerical range assumption, we give a similar construction in \Cref{thm: construction for general numerical range}.

\subsubsection{Improved prefactor for dissipative dynamics} 
For the specific function $f(z)=e^{-z}$ ($\Re(z)\ge 0$), we find a kernel function that improves the cost by a factor of $2.1\times$ over prior art \cite{low2025optimal} and is within a factor of $1.5\times$ of the theoretical lower bound. As in previous LCHS works, our construction also generalizes to the time-dependent ODE
\begin{equation}
    \frac{\rd}{\rd t} \vec{\psi}(t) = -A(t)\vec{\psi}(t),
\end{equation}
which is solved by the time-ordered operator
\begin{equation}\label{eq: time ordered operator def}
    \mathcal{T}e^{-\int_0^t A(s) \rd s}: \vec{\psi}(0) \mapsto \vec{\psi}(t).
\end{equation}

\begin{result}[Informal version of \Cref{thm: exp approx 1}]\label{res: 3}
    If the Cartesian decomposition $A(s) = L(s) + \I H(s)$ satisfies $L(s)\succeq 0$, then the approximate LCHS formula
    \begin{equation}
        \norm{\mathcal{T}e^{-\int_0^t A(s) \rd s}- (2 \pi)^{-1/2} \int_{-R}^R \hat{g}_{R,C}(u) \mathcal{T}e^{-\I\int_0^t \left(u L(s)+H(s)\right)\rd s} \rd u}\le \epsilon
    \end{equation}
    holds for the kernel function
    \begin{align*}
    \hat{g}_{R,C}(u) = \sqrt{\frac{2}{\pi}} e^{2C \arctan\left(\frac{1}{R}\right)} e^{\I C \log\left(\frac{R-u}{R+u}\right)} \frac{1}{1+u^2}
\end{align*}
with parameters $2C \sim R\sim \frac{2}{\pi}\ln\frac{1}{\epsilon}$, and the resulting cost satisfies
$$\Lambda = \lambda_{\hat{g}_{R,C}}R \le \frac{2 e}{\pi}\ln\frac{1}{e\epsilon}.$$
\end{result}

\subsubsection{Direct optimization of discrete LCHS weights} 
We develop a convex optimization framework for finding good discrete weights $\{w_j\}$ in \eqref{eq: disc lchs formula} under a given budget $(R,\lambda)$. The theoretical guarantee is the discrete analogue of \Cref{res: 1}, corresponding to the case where $h$ is a finite sum of Fourier modes. The resulting programs are relatively tight because they optimize directly over the weights $\{w_j\}$, rather than over a parametrized ansatz for the weight coefficients. The costs of the optimized LCHS formulas are typically much smaller than those of the closed-form construction in \Cref{res: 2}. For a fixed approximation region $\Omega$, we use a fixed-stepsize grid $\mathcal{G}$; therefore, the number of terms in \eqref{eq: disc lchs formula} only grows like $\mo{R^2} = \mo{\log^2\frac{1}{\epsilon}}$.

For the one-dimensional interval $[0,T]$ and the two-dimensional regions $[-1,1]^2$ and $D_1 = \{(x,y):x^2+y^2\le 1\}$, we derive concrete local Wiener norm upper bounds and present the corresponding optimization problems. We test the optimization method on several examples.

\subsection{Related work}
The closest line of work is the linear combination of Hamiltonian simulation (LCHS) approach to non-unitary dynamics. The original LCHS method \cite{an2023linear} represents $e^{-tA}$, or the more general time-dependent version \eqref{eq: time ordered operator def}, for $A=L+\I H$ with $L\succeq 0$, as a linear combination of Hamiltonian evolutions using the formula \eqref{eq: LCHS exp(-tA)}, and achieves optimal dependence on the input state preparation cost. The improved LCHS constructions \cite{an2026quantum, pocrnic2025constant, huang2025fourier, low2025optimal} use different kernel functions to obtain better query complexity dependence on the target precision $\epsilon$ for linear non-unitary dynamics. Subsequent works have refined this direction by extending the operator setting \cite{lu2025infinite} and developing problem-specific or implementation-oriented variants \cite{novikau2026efficient,yang2025circuit,das2026quantum}. 

The Lap-LCHS framework \cite{an2026laplace} extends LCHS from matrix exponentials to functions that admit a suitable Laplace representation, thereby covering several eigenvalue transformations such as powers of shifted inverses and mass-matrix differential equations. The core idea is that if a function $f(z)$ has a convergent Laplace transform representation
$$f(z) = \int_0^\infty g(t) e^{-tz} \rd t,\quad \Re(z)\ge 0,$$
then one may implement $f(A)$ as a linear combination of terms $e^{-tA}$, which can be handled by the original LCHS framework. However, only certain analytic functions that decay on the right half-plane admit convergent Laplace transforms. This excludes many common functions, including polynomials $p(A)$, logarithms $\log(A)$, and matrix powers $A^p\ (p>0)$. In contrast, our work treats general analytic functions through a two-dimensional Fourier/Weyl-calculus representation and also allows the weights to be discrete from the outset.

Another general approach to eigenvalue transformation is based on contour integrals. For analytic $f$, Cauchy's formula expresses $f(A)$ as an integral of resolvents $(zI-A)^{-1}$, and quadrature reduces the task to a linear combination of shifted linear-system solves \cite{TakahiraOhashiSogabeEtAl2020,TakahiraOhashiSogabeEtAl2021,tong2021fast}. Recent work gives a more systematic contour-integral-based quantum eigenvalue transformation framework and a detailed complexity analysis \cite{jiang2026contour}. This route is quite general for non-Hermitian and non-normal matrices, but its cost depends on the contour geometry and the implementation of shifted inverses. The query complexity typically deteriorates when $A$ is non-diagonalizable or has a poorly conditioned eigenvector basis \cite{jiang2026contour}, an issue that does not arise in LCHS-based methods.

Polynomial and signal-processing approaches form a complementary family. Quantum singular value transformation \cite{gilyen2019quantum} and related works \cite{dong2022ground, motlagh2024generalized,sunderhauf2023generalized} provide a powerful framework for applying polynomial transformations to singular values, but singular value transformations do not directly give eigenvalue transformations for general non-normal matrices. Quantum eigenvalue processing \cite{low2026quantum} overcomes this obstruction for polynomial eigenvalue transformations using Chebyshev or Faber history states. Its complexity also depends on the condition number of the Jordan basis of $A$, similar to the contour-integral approach. Recent work \cite{gutierrez2026quantum} presents a conceptually simpler eigenvalue-transformation framework based on GQSP \cite{motlagh2024generalized} and the ``regular block encoding'' of the matrix $A$. 


\subsection{Organization of the paper}

The rest of the paper is organized as follows. \Cref{sec:prelim} fixes the Fourier-transform convention and introduces the Wiener norm used to measure LCHS approximation errors. In \Cref{sec:weyl}, we develop the Weyl-calculus framework for two-variable Fourier representations and prove the main LCHS error bound (\Cref{res: 1}). In \Cref{sec:general-analytic-lchs}, we use this bound to construct explicit LCHS formulas for general analytic functions (\Cref{res: 2}). In \Cref{sec:time-ordered-weyl}, we adapt the Weyl-calculus argument to time-ordered non-unitary dynamics and then give the improved kernel function for dissipative dynamics in \Cref{sec:improved-prefactor} (\Cref{res: 3}). In \Cref{sec:convex-optimization}, we present the discrete-weight optimization framework and describe the corresponding convex programs for a one-dimensional interval, a two-dimensional square, and a two-dimensional disc. \Cref{sec:implementation} summarizes coherent and hybrid quantum implementations and explains how the LCHS parameters $(R,\lambda)$ enter the quantum resource estimates. Finally, \Cref{sec:discussion} discusses potential future directions.

\section{Preliminaries: Fourier transforms and Wiener norms}\label{sec:prelim}
We use the following convention for the Fourier transform and its inverse:
\begin{equation*}
g(\bx)=(2 \pi)^{-n / 2} \int_{\RR^n} \hat{g}(\bxi) e^{-\I \bxi \cdot \bx} \rd \bxi
\end{equation*}
\begin{equation*}
\hat{g}(\bxi)=(2 \pi)^{-n / 2} \int_{\RR^n} g(\bx) e^{\I \bxi \cdot \bx} \rd \bx
\end{equation*}
To include discrete LCHS formulas in the same framework, we also use the Fourier--Stieltjes convention. Namely, for a finite complex Borel measure $\mu$, we write $\hat{g} = \mu$ if
\begin{equation*}
g(\bx)=(2 \pi)^{-n / 2} \int_{\mathbb{R}^n} e^{-\I \bxi \cdot \bx} \, \mathrm{d} \mu(\bxi)
\end{equation*}
holds. In particular, $\widehat{e^{-\I \ba \cdot \bx}} = (2 \pi)^{n / 2} \delta_{\ba}$, where $\delta_{\ba}$ is the Dirac measure at $\ba$. We refer the reader to \Cref{appendix: B(Rn)} for more details on the Fourier--Stieltjes algebra. We use the standard notation
\begin{equation}
    A(\RR^n) := \{ g: \hat{g}\in L^1(\RR^n)\}, \quad B(\RR^n) := \{g: \hat{g} = \mu \text{ such that } \int_{\RR^n} \rd \abs{\mu}\le\infty\}.
\end{equation}
For $g\in B(\RR^n)$ and $\mu = \hat{g}$, we define the Wiener norm
\begin{equation}\label{eq: wiener norm def}
    \norm{g}_W := (2 \pi)^{-n / 2} \int_{\RR^n} \rd \abs{\mu} \ \xlongequal{g\in A(\RR^n)} (2 \pi)^{-n / 2} \int_{\RR^n} \abs{\hat{g}(\bxi)} \rd \bxi.
\end{equation}
Under this norm, $B(\RR^n)$ is a Banach space. The space $A(\RR^n)$ is the closed subspace corresponding to absolutely continuous finite measures, since an $L^1$ Fourier transform $\hat{g}(\bxi)$ can be identified with the measure $\mu = \hat{g}(\bxi)\rd \bxi$.

As an important example, when $g$ is a periodic function given by a Fourier series, $\norm{g}_W$ is the $\ell^1$ norm of its Fourier coefficients. 

Suppose that a smooth function $g$ is specified only on a closed set $\Omega \subset \RR^n$, in the sense that it extends to a $C^\infty$ function on a neighborhood of $\Omega$. We define its local Wiener norm by
\begin{equation}\label{eq: restricted wiener norm def}
    \norm{g}_{W(\Omega)} :=\ \inf_{\tilde{g}\in B(\RR^n),\ \tilde{g}|_\Omega = g}\ \norm{\tilde{g}}_W
\end{equation}
When $\Omega$ is compact and convex, the infimum can be taken over more regular extensions without changing its value:
\begin{equation}\label{eq: regular extension}
    \norm{g}_{W(\Omega)} =\ \inf_{\tilde{g}\in B(\RR^n),\ \tilde{g}|_\Omega = g}\ \norm{\tilde{g}}_W =\ \inf_{\hat{h}\in C^{\infty}(\RR^n),\ h\in A(\RR^n),\ h|_\Omega = g}\ \norm{h}_W.
\end{equation}
This equality is proved in \Cref{appendix: local wiener norm}. 

The Wiener norm can often be controlled by Sobolev norms. Roughly speaking, more than $(n/2)$-th weak derivatives are enough to guarantee a finite Wiener norm. In two dimensions, Cauchy--Schwarz gives
\begin{equation}\label{eq: H2 estimate}
        \begin{aligned}
        \norm{g}_{W(\RR^2)} &=\frac{1}{2\pi} \iint_{\RR^2} \abs{\hat{g}(u,v)} \rd u \rd v \\
        &\le \frac{1}{2\pi} \left( \iint_{\RR^2} \frac{1}{(1+u^2+v^2)^2} \rd u \rd v \right)^{1/2} \left( \iint_{\RR^2} (1+u^2+v^2)^2 \abs{\hat{g}(u,v)}^2 \rd u \rd v \right)^{1/2}\\
        &= \frac{1}{2\sqrt{\pi}} \norm{(1-\Delta)g}_{L^2(\RR^2)}.
    \end{aligned}
\end{equation}
Thus, every $g$ in the Sobolev space $H^2(\RR^2)$ has finite Wiener norm. 

Exact local Wiener norms are usually difficult to compute, but useful upper bounds can be obtained from extensions. If $g$ is defined only on a closed set $\Omega\subset\RR^2$, we may extend it to a function in $H^2(\RR^2)$ and then apply \eqref{eq: H2 estimate}. The following standard extension result is sufficient for this purpose.

\begin{lemma}\label{lemma: ext H2}
If $\Omega$ is a compact convex set with non-empty interior in $\RR^n$, then there is a bounded extension operator $E: H^2(\Omega) \to H^2(\RR^n)$.
\end{lemma}
\begin{proof}
By \cite[Corollary 1.2.2.3]{grisvard2011elliptic}, $\Omega$ has Lipschitz boundary. The desired extension theorem for Lipschitz domains is given in \cite[Theorem A.4]{mclean2000strongly}.
\end{proof}

We will also use a cruder $W^{2,\infty}$ bound for the disc $D_a$ of radius $a\ge 1$. The advantage of this estimate is that it gives explicit constants.

\begin{lemma}\label{lemma: err bounded by 2nd derivatives}
    Assume $g\in C^2(D_a)$ with $a\ge 1$, and
    \begin{equation*}
        M_k = \max_{(x,y) \in D_a, |\alpha|=k} |\partial^\alpha g(x,y)|,\qquad k=0,1,2.
    \end{equation*}    
    Then there exists an explicit extension $E_ag$ of $g$ to $\RR^2$ such that
    \begin{equation}\label{eq: unit W bound e_disc}
        \norm{g}_{W(D_a)}\le \norm{E_ag}_{W(\RR^2)}\le a(13 M_0 + 10 M_1 + 8 M_2).
    \end{equation}
\end{lemma}
The construction of $E_ag$ and the proof of the estimate are given in \Cref{appendix: proof of lemma: err bounded by 2nd derivatives}.

\section{Weyl calculus and generalized LCHS}\label{sec:weyl}

Weyl calculus provides a systematic way to lift scalar functions of several variables to functions of noncommuting operators. It was originally developed in connection with the quantization of functions of position and momentum operators.

We only need the two-variable setting, which already contains the essential ideas. Given a scalar function $g(x,y)$ and Hermitian matrices $X,Y$, we want to define an operator-valued function $\op{g}{X,Y}$. For the monomial $p_{m,n}(x,y)=x^m y^n$, a natural definition is
\begin{equation}\label{eq: weyl monomial}
\op{p_{m,n}}{X,Y} = \frac{1}{\binom{m+n}{m}}\sum_{\sigma} A_{\sigma(1)}A_{\sigma(2)}\cdots A_{\sigma(m+n)},
\end{equation}
where the sum is over all binary strings $\sigma$ containing $m$ zeros and $n$ ones, with $A_0=X$ and $A_1=Y$. For example, $\op{p_{2,1}}{X,Y} = \frac{1}{3}(X^2Y+XYX+YX^2)$. We then define $\op{\cdot}{X,Y}$ for polynomials and power series by linearity. With this definition, one checks that
\begin{align}
    \op{(ax+by)^n}{X,Y}  & = (aX+bY)^n,\label{eq: weyl power}\\
    \op{e^{ax+by}}{X,Y} & = e^{aX+bY}.\label{eq: weyl exp}
\end{align} 
Indeed, \eqref{eq: weyl power} follows from the expansion $(ax+by)^n = \sum_{k=0}^n \binom{n}{k} a^kb^{n-k}x^ky^{n-k}$: the operator $\op{x^ky^{n-k}}{X,Y}$ averages over all products with $k$ copies of $X$ and $n-k$ copies of $Y$, matching the corresponding terms in $(aX+bY)^n$. The identity \eqref{eq: weyl exp} then follows by applying \eqref{eq: weyl power} term by term to the power series $e^{ax+by} = \sum_{k=0}^\infty \frac{1}{k!}(ax+by)^k$.

For a function $g$ with $\norm{\hat{g}}_{L^1(\RR^2)} < \infty$, Weyl \cite{weyl1950theory} suggested the definition
\begin{equation}\label{eq: weyl fourier defi}
\op{g}{X,Y} = \frac{1}{2\pi}\iint_{\RR^2} \hat{g}(u,v) e^{-\I(uX+vY)} \rd u \rd v
\end{equation}
The motivation is that
\begin{equation*}
\begin{aligned}
\frac{1}{2\pi}\iint_{\RR^2} \hat{g}(u,v) e^{-\I(uX+vY)} \rd u \rd v 
&= \frac{1}{2\pi}\iint_{\RR^2} \hat{g}(u,v)  \op{e^{-\I(ux+vy)}}{X,Y}\rd u \rd v\\
&= \op{\frac{1}{2\pi}\iint_{\RR^2} \hat{g}(u,v) e^{-\I(ux+vy)} \rd u \rd v}{X,Y}\\
&= \op{g}{X,Y},\\
\end{aligned}
\end{equation*}
where the second equality uses the exponential identity \eqref{eq: weyl exp}.

The relevance of Weyl calculus to LCHS comes from the following observation. Let $f(z)$ be analytic and set $g(x,y)=f(x+\I y)$. If $A=X+\I Y$, then \eqref{eq: weyl power} suggests that $\op{g}{X,Y}=f(X+\I Y)$, since both $g$ and $f$ can be expanded in power series.
At the same time, \eqref{eq: weyl fourier defi} represents $\op{g}{X,Y}$ as a linear combination of unitary matrices $e^{-\I(uX+vY)}$. This is precisely the LCHS form. The main obstacle is that the Fourier transform of an analytic function is usually not an $L^1$ function or a finite Borel measure, so \eqref{eq: weyl fourier defi} cannot be applied directly. The rest of this section develops the framework needed to overcome this obstruction.

We first recall a few functional-analytic notions. The \textit{Schwartz space} $\sch$ consists of smooth functions whose derivatives of all orders decay faster than any polynomial at infinity. A \textit{matrix-valued distribution} is a continuous linear functional from the test function space $C^\infty(\RR^2)$ to the matrix space $\matn$. The \textit{support} of a matrix-valued distribution $T$, denoted by $\supp(T)$, is the complement of the largest open set $U \subset \RR^2$ such that $T(g) = 0$ for all test functions $g$ supported in $U$. 

\begin{lemma}[{\cite[Theorem 2.4 with $s=0$]{jefferies2004spectral}}]
\label{lem: op C infty extension}
Let $X,Y\in\matn$ be Hermitian matrices. Then there exists a unique compactly supported matrix-valued distribution 
$$\op{\cdot}{X,Y}:C^\infty(\RR^2)\to \matn$$
that extends the definition \eqref{eq: weyl monomial} on monomials. Moreover, this $\op{\cdot}{X,Y}$ also satisfies \eqref{eq: weyl fourier defi} for $g\in\sch$. 
\end{lemma}

The lemma states that $\op{\cdot}{X,Y}$ is compactly supported, but it does not identify the support. The next result gives the needed geometric description.

\begin{definition}
    The numerical range of a complex matrix $A\in\matn$ is defined as
    $$W(A) = \{ \bx^\dagger A \bx \mid \bx \in \mathbb{C}^N, \|\bx\|_2 = 1 \}.$$
\end{definition}
\begin{lemma}[{\cite[Theorem 5.2]{anderson1969weyl}}]\label{lem: cvx hull op}
    Let $A = X+\I Y$, and let 
    $K$ denote the support of $\op{\cdot}{X,Y}$. Then $W(A)$ is the closed convex hull of $K$. 
\end{lemma}
By the Toeplitz--Hausdorff theorem, $W(A)$ is a convex subset of $\CC$. It contains all eigenvalues of $A$, but it can be larger than the convex hull of the eigenvalues unless $A$ is normal. Two simple but important examples are
\begin{itemize}
    \item $A+A^\dagger \succeq 0 \Longleftrightarrow W(A)\subset \{z: \Re(z)\ge 0\};$
    \item $\norm{A}\le R \Longrightarrow W(A)\subset \DD_R$.
\end{itemize}

The following proposition gives the connection between eigenvalue transformation and Weyl calculus.

\begin{prop}\label{thm: eigenvalue transform is weyl}
    Let $X,Y$ be Hermitian matrices, and let $\Omega$ denote the numerical range of $A = X+\I Y$. Let $f$ be analytic on an open neighborhood of $\Omega$, and set $g(x,y) = f(x+\I y)$. Then
    \begin{equation}\label{eq: eigenvalue transform is weyl}
        \op{g}{X,Y} = f(X+\I Y).
    \end{equation}
\end{prop}
\begin{proof}
    By \eqref{eq: weyl power} and linearity, for every polynomial $p$ we have
    \begin{equation}\label{eq: polynomial eigenvalue transform is weyl}
        \op{p(x+\I y)}{X,Y}=p(A).
    \end{equation}

    Now let $f$ be analytic on an open neighborhood $U$ of $\Omega$.  Choose a compact convex set $L$ such that
    \[
        \Omega\subset \operatorname{int}(L),\qquad L\subset U .
    \]
    Since $\CC\setminus L$ is connected, Runge's theorem \cite{rudin1974real} gives polynomials $p_m$ such that $p_m\to f$ uniformly on $L$.  By Cauchy's integral formula, the convergence is uniform for all complex derivatives on every compact subset of $\operatorname{int}(L)$, in particular on a neighborhood of $\Omega$.  Thus the functions
    \[
        g_m(x,y)=p_m(x+\I y)
    \]
    converge to $g(x,y)=f(x+\I y)$ in the $C^\infty$ topology on a neighborhood of $\Omega$. By \Cref{lem: cvx hull op}, the distribution $\op{\cdot}{X,Y}$ is supported inside $\Omega$. Together with the continuity of $\op{\cdot}{X,Y}$ in the $C^\infty$ topology, this implies that
    \begin{equation}\label{eq: weyl polynomial limit e_transform}
        \op{g_m}{X,Y}\longrightarrow \op{g}{X,Y}.
    \end{equation}

    On the other hand, the spectrum of $A$ is contained in its numerical range $\Omega$, so
    \begin{equation}\label{eq: matrix polynomial limit e_transform}
        p_m(A)\longrightarrow f(A)
    \end{equation}
    holds by spectral decomposition if $A$ is diagonalizable. The identity follows for general matrices by continuity, since diagonalizable matrices are dense in $\matn$. This completes the proof.
\end{proof}

\Cref{lem: op C infty extension} shows that $\op{g}{X,Y}$ is well-defined for $C^\infty$ functions and depends only on the values of $g$ on a bounded region. However, explicitly computing Fourier transforms of $C^\infty$ functions is usually difficult. We therefore relax the smoothness requirement and allow $g\in B(\RR^2)$. Let
\begin{equation}\label{eq: FS f}
    g(x,y) = (2 \pi)^{-1} \int_{\RR^2} e^{-\I (ux+vy)} \rd \mu(u,v)
\end{equation}
for some finite complex Borel measure $\mu$, and we naturally define
\begin{equation}\label{eq: def opg general}
    \op{g}{X,Y}:=(2 \pi)^{-1} \int_{\RR^2} e^{-\I (uX+vY)} \rd \mu(u,v),
\end{equation}
which is a generalization to \eqref{eq: weyl fourier defi}.
\begin{lemma}\label{lem: op determined by f Omega}
    Let $X,Y$ be Hermitian matrices, and let $\Omega$ denote the numerical range of $X+\I Y$. If $g$ is given by \eqref{eq: FS f}, then $\op{g}{X,Y}$ 
    is determined by $g|_\Omega$. In particular, if $g$ vanishes on $\Omega$, then $\op{g}{X,Y} = 0$.
\end{lemma}
\begin{proof}
By linearity, it suffices to prove that if $g$ vanishes on $\Omega$, then $\op{g}{X,Y} = 0$. 

We first dispose of the degenerate case. If $\Omega$ has empty interior, then, since it is convex, it is contained in an affine line $c+e^{\I\theta}\RR$. Let $B=e^{-\I\theta}(A-cI)$. Then $W(B)\subset \RR$, so $\bx^\dagger \frac{B-B^\dagger}{2\I}\bx=0$ for every unit vector $\bx$. Since $\frac{B-B^\dagger}{2\I}$ is Hermitian, this implies $B=B^\dagger$. Thus $A=cI+e^{\I\theta}B$ for some Hermitian matrix $B$. Writing $c=c_x+\I c_y$ and $e^{\I\theta}=a+\I b$, we have $X=c_xI+aB$ and $Y=c_yI+bB$. Hence
\[
    \op{g}{X,Y}
    =(2\pi)^{-1}\int_{\RR^2} e^{-\I(uc_x+vc_y)}e^{-\I(ua+vb)B}\rd\mu(u,v).
\]
By the spectral theorem, this is the ordinary functional calculus $\phi(B)$ for the one-variable function
\[
    \phi(t)=g(c_x+at,c_y+bt).
\]
Since $W(A)=c+e^{\I\theta}W(B)$ and $g$ vanishes on $\Omega=W(A)$, the function $\phi$ vanishes on $W(B)$. Thus $\phi(B)=0$, because the spectrum of $B$ is contained in $W(B)$. This proves the claim when $\Omega$ has empty interior.

It remains to consider the case where $\Omega$ has nonempty interior. After translating $X$ and $Y$ by suitable scalar multiples of the identity, we may assume without loss of generality that the origin lies in the interior of $\Omega$. The proof has two parts.

\textbf{Part 1: Functions vanishing on an open neighborhood of $\Omega$.}

Suppose that a function $h(x,y) = (2\pi)^{-1}\int_{\RR^2} e^{-\I(ux+vy)} \rd \nu(u,v)$, generated by a finite Borel measure $\nu$, vanishes identically on an open neighborhood $U$ of $\Omega$. We first show that $\op{h}{X,Y} = 0$.

Let $\phi$ be a standard smooth mollifier: it is compactly supported, smooth, nonnegative, and satisfies $\iint_{\RR^2} \phi(x,y) \rd x \rd y = 1$. Define the scaled mollifier $\phi_\epsilon(x,y)=\epsilon^{-2}\phi(x/\epsilon,y/\epsilon)$. The convolution $h_\epsilon=h*\phi_\epsilon$ is smooth. For sufficiently small $\epsilon>0$, it still vanishes on a slightly smaller open neighborhood of $\Omega$. By \Cref{lem: cvx hull op}, the support of the distribution $\op{\cdot}{X,Y}$ is contained in $\Omega$. Therefore $\op{h_\epsilon}{X,Y}=0$.

In the frequency domain, the Fourier transform of $h_{\epsilon}=h*\phi_\epsilon$ is the measure $2\pi\hat{\phi}_\epsilon\nu$. Since $\hat{\phi}_\epsilon(u,v)=(2\pi)^{-1}\iint_{\RR^2}\phi_\epsilon(x,y)e^{\I(ux+vy)}\rd x\rd y$, we have $\hat{\phi}_\epsilon(u,v)=\hat{\phi}(\epsilon u,\epsilon v)$. Hence
\begin{align*} 
0 &= \op{h_\epsilon}{X,Y} = (2\pi)^{-1} \int_{\RR^2} e^{-\I(uX+vY)} 2\pi\hat{\phi}_\epsilon(u,v) \rd \nu(u,v)\\
&= (2\pi)^{-1} \int_{\RR^2} e^{-\I(uX+vY)} 2\pi\hat{\phi}(\epsilon u, \epsilon v) \rd \nu(u,v). 
\end{align*}
As $\epsilon \to 0$, the factor $2\pi\hat{\phi}(\epsilon u,\epsilon v)$ converges pointwise to $2\pi\hat{\phi}(0,0)=\iint_{\RR^2}\phi(x,y)\rd x\rd y=1$. Since the matrix exponential is unitary and $\nu$ is finite, the dominated convergence theorem gives $0=\op{h}{X,Y}$.

\textbf{Part 2: The function $g$ vanishes exactly on $\Omega$.}

Now suppose that $g$ vanishes on the closed set $\Omega$. For any $R>1$, define the dilated function $g_R(x,y)=g(x/R,y/R)$. Since $g$ vanishes on $\Omega$ and the origin is in the interior of $\Omega$, the function $g_R$ vanishes on the scaled set $R\Omega=\{(Rx,Ry):(x,y)\in\Omega\}$. For $R>1$, this set contains an open neighborhood of $\Omega$, because $\Omega$ is convex by the Toeplitz--Hausdorff theorem.

The function $g_R$ corresponds to the scaled measure $\mu_R$ defined by $\int \psi(u,v) \rd \mu_R(u,v) = \int \psi(u/R, v/R) \rd \mu(u,v)$ for test functions $\psi$. This scaling preserves total variation, so $\int \rd \abs{\mu_R} = \int \rd\abs{\mu} < \infty$. More explicitly,
$$ g_R(x,y) = (2\pi)^{-1} \int_{\RR^2} e^{-\I(ux+vy)} \rd \mu_R(u,v) = (2\pi)^{-1} \int_{\RR^2} e^{-\I(\frac{u}{R}x + \frac{v}{R}y)} \rd \mu(u,v). $$
Since $g_R$ vanishes on an open neighborhood of $\Omega$, Part 1 gives $\op{g_R}{X,Y}=0$. Hence
$$ 0 = \op{g_R}{X,Y} = (2\pi)^{-1} \int_{\RR^2} e^{-\I(\frac{u}{R}X + \frac{v}{R}Y)} \rd \mu(u,v). $$
Finally, let $R\to1^+$. The matrix exponential $e^{-\I(\frac{u}{R}X+\frac{v}{R}Y)}$ converges pointwise to $e^{-\I(uX+vY)}$, and its operator norm is bounded by $1$. Another application of the dominated convergence theorem yields
$$ \op{g}{X,Y} = (2\pi)^{-1} \int_{\RR^2} e^{-\I(uX+vY)} \rd \mu(u,v) = \lim_{R \to 1^+} (2\pi)^{-1} \int_{\RR^2} e^{-\I(\frac{u}{R}X + \frac{v}{R}Y)} \rd \mu(u,v) = 0. $$
This completes the proof.
\end{proof}

The preceding lemma shows that, for a function $g$ defined on $\Omega$, the operator $\op{g}{X,Y}$ is well-defined whenever $g$ admits an extension with finite Wiener norm. We will often approximate $g$ by another function, and the next estimate controls the resulting operator error.

\begin{cor}\label{thm: opf bound}
    Let $X,Y$ be Hermitian matrices, and let $\Omega$ denote the numerical range of $X+\I Y$. If $g$ is a smooth function defined on $\Omega$, then
    $$\norm{\op{g}{X,Y}}\le \norm{g}_{W(\Omega)}.$$ 
\end{cor}
\begin{proof}
    For any function $\tilde{g}\in B(\RR^2)$ that extends $g$ to $\RR^2$ and satisfies $(\tilde{g})^\wedge = \tilde{\mu}$, \Cref{lem: op determined by f Omega} gives $\op{g}{X,Y} = \op{\tilde{g}}{X,Y}$. By definition \eqref{eq: def opg general},
    $$ \norm{\op{\tilde{g}}{X,Y}}=\norm{(2 \pi)^{-1} \int_{\RR^2} e^{-\I (uX+vY)} \rd \tilde{\mu}(u,v)}\le (2 \pi)^{-1} \int_{\RR^2}  \rd \abs{\tilde{\mu}}\le \norm{\tilde{g}}_{W}.$$ 
    Taking the infimum over all extensions $\tilde{g}$ completes the proof by the definition of $\norm{g}_{W(\Omega)}$.
\end{proof}

\begin{theorem}\label{thm: f - h bound}
    Let $X,Y$ be Hermitian matrices, and let $\Omega$ denote the numerical range of $X+\I Y$. Let $f$ be analytic on an open neighborhood of $\Omega$, and set $g(x,y) = f(x+\I y)$. Then, for any $h\in B(\RR^2)$,
    \begin{equation}\label{eq: f - h bound}
        \norm{f(X+\I Y) - \op{h}{X,Y}} \le \norm{g - h}_{W(\Omega)}.
    \end{equation}
\end{theorem}
\begin{proof}
    Applying \Cref{thm: opf bound} to $g-h$ gives
\begin{equation*}
    \norm{\op{g}{X,Y} - \op{h}{X,Y}} \le \norm{g-h}_{W(\Omega)}.
\end{equation*}
By \Cref{thm: eigenvalue transform is weyl}, $\op{g}{X,Y} = f(X+\I Y)$, which proves \eqref{eq: f - h bound}.
\end{proof}

This theorem is the central estimate in the Weyl-calculus approach to LCHS. It reduces the construction of approximate LCHS formulas for an analytic function $f$ to a scalar approximation problem: one only needs to approximate the two-variable function $f(x+\I y)$ on the numerical range of $A$. Since the Fourier--Stieltjes transform of $h$ may be any finite Borel measure, the same statement covers both continuous and discrete LCHS formulas.

\section{LCHS-based eigenvalue transformation for analytic functions}\label{sec:general-analytic-lchs}
In this section, we use the Weyl-calculus error bound to construct explicit LCHS formulas for general analytic functions. The construction depends on a region containing the numerical range of the input matrix $A$. We first treat the common normalization $\norm{A}\le 1$, which is especially natural in the block-encoding query model (\Cref{def: block encoding}). Under this assumption, $W(A)\subset \DD_1$. The more general case $\norm{A}\le \alpha$ can be reduced by scaling. In addition, we will also state a version for arbitrary bounded convex numerical ranges at the end of the section.

Assume that $f(z)$ is analytic in a neighborhood of $\DD_1$. Then there exists $\rho>0$ such that $f$ is analytic and bounded in the interior of $\DD_a$, where $a=1+\rho$. Define the two-variable function $g:\RR^2\to \CC$ by $g(x,y)=f(x+\I y)$. Our goal is to construct a smooth approximation to $g$ using only local data on $D_a$.

We fix an extension $g_2$ of $g|_{D_a}$ whose Fourier transform is an $L^1$ function:
\begin{equation}\label{eq: g2 L1 extension e_disc}
    g_2=g\quad \text{on }D_a,
    \qquad \hat g_2\in L^1(\RR^2).
\end{equation}
At this stage, the construction works with any such extension $g_2$. Specific choices of $g_2$ will be used later in \Cref{thm: construction 1 for disc} and \Cref{thm: construction 2 for disc}.

Put
\begin{equation}\label{eq: rho and M e_disc}
    \rho=a-1,\qquad
    M=\norm{g_2}_W=(2\pi)^{-1}\norm{\hat g_2}_{L^1(\RR^2)}.
\end{equation}
Then
\begin{equation}\label{eq: g2 Linfty M e_disc}
    \norm{g_2}_{L^\infty(\RR^2)}\le M.
\end{equation}

We convolve $g_2$ with the Gaussian kernel
\begin{equation*}
    K_\gamma(x,y)=\frac{\gamma^2}{\pi}e^{-\gamma^2(x^2+y^2)}.
\end{equation*}
This gives the analytic function
\begin{equation}\label{eq: g3 e_disc}
    g_3(x,y) = K_\gamma*g_2(x,y).
\end{equation}
Equivalently, in the Fourier domain,
\begin{equation}\label{eq: g3 Fourier e_disc}
    \hat g_3(\xi,\eta)=e^{-\frac{\xi^2+\eta^2}{4\gamma^2}}\hat g_2(\xi,\eta).
\end{equation}
Given an accuracy target $\epsilon$, we will choose $\gamma$ and $R$ so that $g_3$ satisfies the following three properties:
\begin{itemize}
    \item $\norm{g_3}_{W}$ is bounded uniformly for $\epsilon$.
    \item $\norm{g_3-g}_{W(D_1)} \le \frac{\epsilon}{2}$.
    \item $(2\pi)^{-1}\norm{\hat{g_3}}_{L^1(\RR^2\setminus D_R)}\le\frac{\epsilon}{2}$ and $R = \mo{\log\frac{1}{\epsilon}}$.
\end{itemize}
Together, these properties allow us to apply \Cref{thm: f - h bound} with the truncated kernel
$\hat h=\hat g_3|_{D_R}$, yielding an LCHS formula with controlled approximation error and cost.
The next three lemmas make the choices of $\gamma$ and $R$ precise.

\begin{lemma}\label{lem: g3 W bounded e_disc} Assume $a>1$. Then
\begin{equation}\label{eq: g3 W minimal e_disc}
    \norm{g_3}_W\le\norm{g_2}_W=M.
\end{equation}
In particular, $\norm{g_3}_W$ is bounded uniformly in $\gamma$.
\end{lemma}
\begin{proof}
By \eqref{eq: g3 Fourier e_disc},
\begin{equation*}
    \hat g_3(\xi,\eta)=e^{-\frac{\xi^2+\eta^2}{4\gamma^2}}\hat g_2(\xi,\eta).
\end{equation*}
Therefore
\begin{equation*}
    \norm{g_3}_W
    =(2\pi)^{-1}\iint_{\RR^2}e^{-\frac{\xi^2+\eta^2}{4\gamma^2}}\abs{\hat g_2(\xi,\eta)}\rd\xi\rd\eta
    \le (2\pi)^{-1}\norm{\hat g_2}_{L^1(\RR^2)}=\norm{g_2}_W.
\end{equation*}
\end{proof}

\begin{lemma}\label{lem: local gaussian approximation e_disc}
    Assume $\rho>0$ and $\gamma\rho\ge 1$. Set $M=\norm{g_2}_W$. Then
    \begin{equation}\label{eq: H2 h bound final e_disc}
        \norm{g_3-g}_{W(D_1)}
        \le \frac{M}{(1\wedge \rho^2)}(62+94\gamma^2\rho^2+56\gamma^4\rho^4)e^{-\gamma^2\rho^2},
    \end{equation}
    where $1\wedge \rho^2 = \min\{1,\rho^2\}$. Consequently, for any $\epsilon \in (0,M)$, the choice
    \begin{equation}\label{eq: gamma choice local e_disc}
        \gamma = \frac{1}{\rho}\sqrt{2\ln\left(\frac{380M}{(1\wedge \rho^2)\epsilon}\right)},
    \end{equation}
    gives
    $$\norm{g_3-g}_{W(D_1)}\le \frac{\epsilon}{2}.$$
\end{lemma}
\begin{proof}
Let
\begin{equation*}
    h(x,y)=g_3(x,y)-g(x,y),\qquad
    a = \rho + 1,
    \qquad t=\gamma^2\rho^2,
    \qquad E=e^{-t}.
\end{equation*}

The function $g_2$ is bounded by \eqref{eq: g2 Linfty M e_disc}, and $\partial^\alpha K_\gamma\in L^1(\RR^2)$ for every multi-index $\alpha$. Dominated convergence justifies differentiating under the integral sign:
\begin{equation}\label{eq: differentiate gaussian e_disc}
    \partial^\alpha g_3(x,y)
    =\iint_{\RR^2}\partial^\alpha_{x,y}
    K_\gamma(x-p,y-q)g_2(p,q)\rd p\rd q.
\end{equation}

For $(x,y)\in D_1$, split the integral in \eqref{eq: differentiate gaussian e_disc} into the local part
\begin{equation*}
    I_\alpha(x,y)=\iint_{(p-x)^2+(q-y)^2\le \rho^2}
    \partial^\alpha_{x,y}K_\gamma(x-p,y-q)g_2(p,q)\rd p\rd q
\end{equation*}
and the tail part
\begin{equation}\label{eq: tail part e_disc}
    T_\alpha(x,y)=\iint_{(p-x)^2+(q-y)^2> \rho^2}
    \partial^\alpha_{x,y}K_\gamma(x-p,y-q)g_2(p,q)\rd p\rd q.
\end{equation}
On the local part, $g_2=g$. We now compute $I_\alpha$ using the harmonic Fourier expansion of $g$ around $(x,y)$. Write
\[
    p=x+s\cos\theta,\qquad q=y+s\sin\theta,\qquad 0\le s\le \rho.
\]
Since $g$ is harmonic in this disc, it has the expansion
\begin{equation}\label{eq: harmonic expansion local e_disc}
\begin{aligned}
    g(x+s\cos\theta,y+s\sin\theta)
    &=g(x,y)+s\big(\partial_xg(x,y)\cos\theta+\partial_yg(x,y)\sin\theta\big)\\
    &\quad+\frac{s^2}{2}\partial_{xx}g(x,y)\cos 2\theta
    +\frac{s^2}{2}\partial_{xy}g(x,y)\sin 2\theta\\
    &\quad+\sum_{m\ge 3}s^m(a_m\cos m\theta+b_m\sin m\theta).
\end{aligned}
\end{equation}
Here the identity $\partial_{yy}g=-\partial_{xx}g$ is used in the second-order term. The expansion converges uniformly on every smaller concentric disc, and the identities below follow by first integrating over a radius $\rho'<\rho$ and then taking $\rho'\uparrow\rho$.

The angular factors coming from $\partial^\alpha K_\gamma$ have order at most two, so orthogonality eliminates all modes in \eqref{eq: harmonic expansion local e_disc} except the zeroth, first, and second modes. More explicitly,
\begin{equation*}
\begin{aligned}
    \partial_xK_\gamma(x-p,y-q)&=2\gamma^2s\cos\theta\,K_\gamma(s),\\
    \partial_yK_\gamma(x-p,y-q)&=2\gamma^2s\sin\theta\,K_\gamma(s),\\
    \partial_{xy}K_\gamma(x-p,y-q)&=4\gamma^4s^2\cos\theta\sin\theta\,K_\gamma(s),\\
    \partial_{jj}K_\gamma(x-p,y-q)&=\big(4\gamma^4s^2 e_j(\theta)^2-2\gamma^2\big)K_\gamma(s),
\end{aligned}
\end{equation*}
where $K_\gamma(s)=\frac{\gamma^2}{\pi}e^{-\gamma^2s^2}$, $e_x(\theta)=\cos\theta$, and $e_y(\theta)=\sin\theta$.  Using the angular integrals
\[
    \int_0^{2\pi}\cos^2\theta\rd\theta
    =\int_0^{2\pi}\sin^2\theta\rd\theta=\pi,\qquad
    \int_0^{2\pi}\cos^2\theta\sin^2\theta\rd\theta=\frac{\pi}{4},
\]
one obtains, for $j,k\in\{x,y\}$ and $j\ne k$,
\begin{align}
    I_0(x,y)&=(1-E)g(x,y),\label{eq: local moment 0 e_disc}\\
    I_j(x,y)&=\left(1-(1+t)E\right)\partial_j g(x,y),\label{eq: local moment 1 e_disc}\\
    I_{jk}(x,y)&=\left(1-\left(1+t+\frac{t^2}{2}\right)E\right)\partial_{jk}g(x,y),\label{eq: local moment mixed e_disc}\\
    I_{jj}(x,y)&=\left(1-\left(1+t+\frac{t^2}{2}\right)E\right)\partial_{jj}g(x,y)-2\gamma^2tE\,g(x,y).
    \label{eq: local moment diagonal e_disc}
\end{align}
The radial integrals used here are
\begin{align*}
    2\gamma^2\int_0^\rho se^{-\gamma^2s^2}\rd s
    &=1-E,\\
    2\gamma^4\int_0^\rho s^3e^{-\gamma^2s^2}\rd s
    &=1-(1+t)E,\\
    \gamma^6\int_0^\rho s^5e^{-\gamma^2s^2}\rd s
    &=1-\left(1+t+\frac{t^2}{2}\right)E.
\end{align*}
For the diagonal second derivatives, the zeroth angular mode contributes the extra term
\begin{equation*}
    g(x,y)\int_0^\rho\int_0^{2\pi}
    \big(4\gamma^4s^2e_j(\theta)^2-2\gamma^2\big)K_\gamma(s)s\rd\theta\rd s
    =-2\gamma^2tE\,g(x,y).
\end{equation*}

Next we estimate the tails in \eqref{eq: tail part e_disc}. By \eqref{eq: g2 Linfty M e_disc}, we may bound $\abs{g_2}$ by $M$.
For $s=\sqrt{(p-x)^2+(q-y)^2}$, the Gaussian derivative bounds are
\begin{equation}\label{eq: gaussian derivative bounds local e_disc}
    \abs{\partial^\alpha_{x,y}K_\gamma(x-p,y-q)}
    \le
    \begin{cases}
        K_\gamma(x-p,y-q),& |\alpha|=0,\\[0.4ex]
        2\gamma^2sK_\gamma(x-p,y-q),& |\alpha|=1,\\[0.4ex]
        (4\gamma^4s^2+2\gamma^2)K_\gamma(x-p,y-q),& |\alpha|=2.
    \end{cases}
\end{equation}
Using polar coordinates centered at $(x,y)$, we obtain
\begin{align}
    |T_0(x,y)|&\le ME,\label{eq: tail zero e_disc}\\
    |T_j(x,y)|&\le M\left(2\gamma^2\rho+\frac{1}{\rho}\right)E,
    \qquad j\in\{x,y\},\label{eq: tail first e_disc}\\
    |T_\alpha(x,y)|&\le M\,2\gamma^2(2t+3)E,
    \qquad |\alpha|=2.\label{eq: tail second e_disc}
\end{align}
Indeed, the first derivative estimate uses
\begin{equation*}
    4\gamma^4\int_\rho^\infty s^2e^{-\gamma^2s^2}\rd s = 2\gamma^2\rho e^{-\gamma^2\rho^2} + 2\gamma^2\int_\rho^\infty e^{-\gamma^2s^2}\rd s
    \le \left(2\gamma^2\rho+\frac{1}{\rho}\right)E,
\end{equation*}
and the second derivative estimate uses
\begin{equation*}
    \iint_{s>\rho}s^2K_\gamma(s)\rd p\rd q=(\rho^2+\gamma^{-2})E.
\end{equation*}

Let
\begin{equation*}
    M'_1=\max\left\{\norm{\partial_xg}_{L^\infty(D_1)},\norm{\partial_yg}_{L^\infty(D_1)}\right\},
    \qquad
    M'_2=\max_{|\alpha|=2}\norm{\partial^\alpha g}_{L^\infty(D_1)}.
\end{equation*}
Combining \eqref{eq: local moment 0 e_disc}--\eqref{eq: tail second e_disc}, we obtain the pointwise bounds for $(x,y)\in D_1$:
\begin{align}
    |h(x,y)|&\le 2ME,\label{eq: h pointwise zero e_disc}\\
    |\partial_jh(x,y)|&\le \left[M\left(2\gamma^2\rho+\frac{1}{\rho}\right)+(1+t)M'_1\right]E,
    \qquad j\in\{x,y\},\label{eq: h pointwise first e_disc}\\
    |\partial^\alpha h(x,y)|&\le \left[M\left(2\gamma^2(2t+3)+2\gamma^2t\right)+\left(1+t+\frac{t^2}{2}\right)M'_2\right]E,
    \qquad |\alpha|=2.\label{eq: h pointwise second e_disc}
\end{align}

It remains to bound $M'_1$ and $M'_2$ in terms of $M$. Recall that $g(x,y)=f(z)$ for $z=x+\I y\in \DD_a$. Then $f'(z)=\partial_x g(x,y)=-\I \partial_yg(x,y)$ and $f''(z)=\partial^2_{xx}g=-\I \partial^2_{xy}g=-\partial^2_{yy}g$. Therefore,
$$M_1' = \max_{z\in \DD}|f'(z)|,\qquad M_2' = \max_{z\in \DD}|f''(z)|.$$
Since $\max_{z\in \DD_a} |f(z)| = \norm{g}_{L^\infty(D_a)}\le M$, Cauchy's estimates for analytic functions, $\abs{f^{(k)}(z)}\le \frac{k!}{\rho^k}\max_{\abs{w}=\rho}\abs{f(z+w)}$, imply
$$M_1' \le \frac{M}{\rho},\qquad M_2' \le \frac{2M}{\rho^2}.$$

Because $t=\gamma^2\rho^2\ge 1$, the estimates \eqref{eq: h pointwise zero e_disc}--\eqref{eq: h pointwise second e_disc} imply the uniform bounds
\begin{align*}
    |h(x,y)|&\le 2ME,\\
    |\partial_jh(x,y)|&\le \frac{1}{\rho}(3t+2) ME,
    \qquad j\in\{x,y\},\\
    |\partial^\alpha h(x,y)|&\le \frac{1}{\rho^2}(7t^2+8t+2)ME,
    \qquad |\alpha|=2.
\end{align*}
Using \Cref{lemma: err bounded by 2nd derivatives}, we obtain
\begin{equation}
    \norm{h}_{W(D_1)}\le \frac{1}{(1\wedge \rho^2)}ME(13\cdot2+10(3t+2)+8(7t^2+8t+2))= \frac{1}{(1\wedge \rho^2)}(62+94t+56t^2)ME
\end{equation}
This proves \eqref{eq: H2 h bound final e_disc}.



Finally, the elementary inequality
\begin{equation*}
    (62+94t+56t^2)e^{-t}\le 190e^{-t/2},
    \qquad t\ge0,
\end{equation*}
gives
\begin{equation*}
    \norm{g_3-g}_{W(D_1)}
    \le 190\frac{M}{(1\wedge \rho^2)}e^{-t/2}.
\end{equation*}
Therefore, the choice \eqref{eq: gamma choice local e_disc} gives $t\ge1$ and $\norm{g_3-g}_{W(D_1)} \le\epsilon/2$.
\end{proof}

\begin{lemma}\label{lem: fourier tail e_disc}
For any $\epsilon\in (0,M)$, choose
\begin{equation}\label{eq: R choice e_disc}
    R=2\gamma\sqrt{\ln\left(\frac{2 M}{\epsilon}\right)}
\end{equation}
Then $(2\pi)^{-1}\norm{\hat{g_3}}_{L^1(\RR^2\setminus D_R)}\le\frac{\epsilon}{2}$.
\end{lemma}
\begin{proof}
By \eqref{eq: g3 Fourier e_disc},
\begin{equation*}
    \hat g_3(\xi,\eta)=e^{-\frac{\xi^2+\eta^2}{4\gamma^2}}\hat g_2(\xi,\eta),
    \qquad \norm{\hat g_2}_{L^1(\RR^2)}=2\pi M.
\end{equation*}
Therefore, for any $R>0$,
\begin{align*}
    \norm{\hat{g_3}}_{L^1(\RR^2\setminus D_R)}
    &=\iint_{\xi^2+\eta^2>R^2}e^{-\frac{\xi^2+\eta^2}{4\gamma^2}}\abs{\hat g_2(\xi,\eta)}\rd\xi\rd\eta\\
    &\le \norm{\hat g_2}_{L^1(\RR^2)}\sup_{\xi^2+\eta^2>R^2}e^{-\frac{\xi^2+\eta^2}{4\gamma^2}} = 2\pi M e^{-\frac{R^2}{4\gamma^2}}.
\end{align*}
With the choice \eqref{eq: R choice e_disc},
\begin{equation*}
    e^{-\frac{R^2}{4\gamma^2}}
    =\frac{\epsilon}{2 M}.
\end{equation*}
Hence
\begin{equation*}
    (2\pi)^{-1}\norm{\hat{g_3}}_{L^1(\RR^2\setminus D_R)}
    \le  \frac{\epsilon}{2}.
\end{equation*}
\end{proof}

\begin{theorem}\label{thm: construction 1 for disc}
    Assume $X, Y$ are Hermitian matrices with $\norm{X+\I Y}\le 1$. Let $\rho >0$, and suppose that $f$ is analytic in the interior of disc $\DD_{1+\rho}$. Define the two-variable function $g(x,y)=f(x+\I y)$ on $D_{1+\rho}$, and assume
    $$\norm{g}_{W(D_{1+\rho})} = B.$$
    Then there exist $R>0$ and a smooth function $\hat{h}:D_R\to \CC$ such that
    \begin{equation}\label{eq: f approx by h}
        \norm{f(X+\I Y) - (2\pi)^{-1}\int_{u^2+v^2\le R^2} \hat{h}(u,v) e^{-\I(uX+vY)} \rd u \rd v}\le \epsilon 
    \end{equation}
    with the cost function
    $$\Lambda = \lambda_{\hat{h}} R \le \frac{3B}{\rho}\ln\frac{30B}{(\rho\wedge1)\epsilon}.$$
\end{theorem}
\begin{proof}
    Let $a=1+\rho$.
    By the definition of the local Wiener norm and \eqref{eq: regular extension}, there exists an extension $g_2$ of $g(x,y)=f(x+\I y)$ such that $g_2=g$ on $D_a$, $\hat{g}_2\in C^{\infty}(\RR^2)\cap L^1(\RR^2)$, and $\norm{g_2}_{W}=M\le 1.01 B$. Let $g_3 = K_\gamma*g_2$, with $\gamma$ chosen according to \eqref{eq: gamma choice local e_disc}. We choose $R$ as in \eqref{eq: R choice e_disc}, and define the truncated Fourier weight $\hat{h} = \hat{g}_3|_{D_R}$. Applying \Cref{lem: local gaussian approximation e_disc} and \Cref{lem: fourier tail e_disc} gives
    $$\norm{g-h}_{W(D_1)} \le \epsilon,$$ 
    where $h$ is the inverse Fourier transform of $\hat h$. Since $\norm{X+\I Y}\le 1$, the numerical range of $X+\I Y$ is contained in $D_1$. By \Cref{thm: f - h bound}, we have 
\begin{equation*}
    \norm{f(X+\I Y) - \op{h}{X,Y}} \le \norm{g-h}_{W(D_1)} \le \epsilon.
\end{equation*}
By definition,
\begin{equation*}
    \op{h}{X,Y} = (2\pi)^{-1}\int_{u^2+v^2\le R^2} \hat{h}(u,v) e^{-\I(uX+vY)} \rd u \rd v.
\end{equation*}
This proves \eqref{eq: f approx by h}.

Finally, we estimate the cost function $\Lambda$. By \Cref{lem: g3 W bounded e_disc}, \Cref{lem: local gaussian approximation e_disc}, and \Cref{lem: fourier tail e_disc}, we obtain
$$\Lambda = \lambda_{\hat{h}} R \le \norm{g_3}_{W} R = \frac{2\sqrt{2}M}{\rho}\sqrt{\ln\left(\frac{380M}{(1\wedge \rho^2)\epsilon}\right)\ln\left(\frac{2 M}{\epsilon}\right)} \le \frac{2\sqrt{2}M}{\rho}\ln\frac{\sqrt{760}M}{(\rho\wedge1)\epsilon},$$
where the last step uses $\sqrt{\ln x \ln y}\le \frac{\ln x + \ln y}{2} = \ln\sqrt{xy}$. Substituting $M\le 1.01 B$ gives the stated bound for $\Lambda$.
\end{proof}

\begin{remark}
    We stated the results above for the case $\hat g_2\in L^1(\RR^2)$. The same proof also works if $\hat g_2(\xi,\eta)\rd\xi\rd\eta$ is replaced by a finite complex Borel measure $\mu_2$: the Fourier-domain formulas are read in the Fourier--Stieltjes sense, $\norm{\hat g_3}_{L^1(\RR^2\setminus D_R)}$ is replaced by the total variation of $e^{-(\xi^2+\eta^2)/(4\gamma^2)}\mu_2$ on $\RR^2\setminus D_R$, and the spatial estimates use only the bound $\norm{g_2}_{L^\infty}\le (2\pi)^{-1}|\mu_2|(\RR^2)$. Indeed, $\norm{g_2}_W = B$ is achievable for some $g_2\in B(\RR^2)$ by \Cref{lem: attaining local W appendix}.
\end{remark}

The preceding theorem is not fully explicit because the extension $g_2$ is not explicitly constructed. In principle, one can compute such an extension numerically, since finding a minimum-Wiener-norm extension is a convex optimization problem. An alternative choice is to use the explicit extension $g_2=E_ag$ from \Cref{lemma: err bounded by 2nd derivatives}. This gives a cost bound in terms of derivative bounds for $f$ up to second order, which we state in the following theorem.

\begin{theorem}\label{thm: construction 2 for disc}
    Assume $X, Y$ are Hermitian matrices with $\norm{X+\I Y}\le 1$. Let $\rho >0$, and let $f$ be analytic in the interior of disc $\DD_{1+\rho}$, satisfying
    $$\max_{\abs{z}\le 1+\rho} \left(\max\{\abs{f(z)}, \abs{f'(z)},\abs{f''(z)}\}\right)\le S $$
    for some $S$. Then there exist $R>0$ and a smooth function $\hat{h}:D_R\to \CC$ such that
    \begin{equation}\label{eq: f approx by h explicit e_disc}
    \norm{f(X+\I Y) - (2\pi)^{-1}\int_{u^2+v^2\le R^2} \hat{h}(u,v) e^{-\I(uX+vY)} \rd u \rd v}\le \epsilon 
    \end{equation}
    with the cost function
    $$\Lambda = \lambda_{\hat{h}} R = \mo{\frac{(\rho+1)S}{\rho}\log\frac{(\rho+1)S}{(\rho\wedge1)\epsilon}}.$$
\end{theorem}
\begin{proof}
Let $a=1+\rho$ and $g(x,y)=f(x+\I y)$. By the assumptions on $f$,
\[
    M_k=\max_{(x,y)\in D_a,\ |\alpha|=k}\abs{\partial^\alpha g(x,y)}\le S,
    \qquad k=0,1,2,
\]
because the Cartesian derivatives of $g$ up to second order are given by $f'$ and $f''$ up to factors in $\{1,\I,-1,-\I\}$. Therefore, we can use the same extension as in \Cref{lemma: err bounded by 2nd derivatives} to bound the local Wiener norm of $g$. For completeness, we record the extension. Define
\begin{equation*}
    G(r,\theta)=
    \begin{cases}
        g(r,\theta), & 0\le r\le a,\\
        \frac32 g(a,\theta)-\frac12 g(3a-2r,\theta), & a<r\le 2a,
    \end{cases}
\end{equation*}
and the $C^1$ cutoff function
\begin{equation*}
    \chi(t)=
    \begin{cases}
        1, & 0\le t\le 1,\\
        1-2(t-1)^2, & 1<t\le \frac32,\\
        2(2-t)^2, & \frac32<t\le 2,\\
        0, & t>2.
    \end{cases}
\end{equation*}
Define the extension $g_2$ as the following $C^1$ function supported on $D_{2a}$:
\begin{equation}
    g_2(r,\theta) = G(r,\theta)\chi(r/a).
\end{equation}
This is the same extension used in the proof of \Cref{lemma: err bounded by 2nd derivatives}; hence
\[
    M:=\norm{g_2}_W\le a(13S+10S+8S)=31(1+\rho)S.
\]

The rest of the construction is the same as in the previous theorem. Let
$g_3=K_\gamma*g_2$, choose $\gamma$ as in \eqref{eq: gamma choice local e_disc}, choose $R$ as in \eqref{eq: R choice e_disc}, and define $\hat{h} = \hat{g}_3|_{D_R}$.
By \Cref{lem: local gaussian approximation e_disc} and \Cref{lem: fourier tail e_disc},
\[
    \norm{g-h}_{W(D_1)}\le\epsilon,
\]
where $h$ is the inverse Fourier transform of $\hat h$. Since $\norm{X+\I Y}\le 1$, the numerical range of $X+\I Y$ is contained in $D_1$, so \Cref{thm: f - h bound} gives
\[
    \norm{f(X+\I Y)-\op{h}{X,Y}}\le \norm{g-h}_{W(D_1)}\le \epsilon.
\]
Together with the definition of $\op{h}{X,Y}$, we prove \eqref{eq: f approx by h explicit e_disc}.

Finally, we estimate the cost function as in the previous theorem:
\[
\begin{aligned}
    \Lambda=\lambda_{\hat h}R
    \le \norm{g_3}_W R
    \le \frac{2\sqrt{2}M}{\rho}
    \sqrt{\ln\left(\frac{380M}{(1\wedge\rho^2)\epsilon}\right)
    \ln\left(\frac{2M}{\epsilon}\right)}
    \le \frac{2\sqrt{2}M}{\rho}
    \ln\frac{\sqrt{760}M}{(1\wedge\rho)\epsilon}.
\end{aligned}
\]
Substituting $M\le31(1+\rho)S$ gives the stated bound
\[
    \Lambda
    =\frac{90(\rho+1)S}{\rho}
    \ln\left(\frac{900(\rho+1)S}{(\rho\wedge1)\epsilon}\right).
\]
\end{proof}

\Cref{thm: construction 1 for disc,thm: construction 2 for disc} were stated for the case where the numerical range of $A=X+\I Y$ is contained in the unit disc $D_1$. The same construction extends to other numerical ranges. Let $\Omega$ be the numerical range of $A$, and let $\Omega_\rho$ be the $\rho$-neighborhood of $\Omega$. If $f$ is analytic in $\Omega_\rho$, we can again choose an extension $g_2$ with finite Wiener norm and convolve it with a Gaussian kernel. The proofs of \Cref{lem: g3 W bounded e_disc} and \Cref{lem: fourier tail e_disc} are unchanged. In \Cref{lem: local gaussian approximation e_disc}, the only geometric property of $D_1$ used is that every ball of radius $\rho$ centered in $D_1$ is contained in $D_{1+\rho}$. The corresponding property holds for any compact set $\Omega$ and its $\rho$-neighborhood $\Omega_\rho$. Finally, the use of \Cref{lemma: err bounded by 2nd derivatives} can be replaced by the extension result \Cref{lemma: ext H2}, which only changes the constants. This gives the following theorem for general numerical ranges.

\begin{theorem}\label{thm: construction for general numerical range}
    Let $X,Y$ be Hermitian matrices, and let $\Omega$ be the numerical range of $X+\I Y$. Let $0< \rho \le 1$, and let $\Omega_{\rho}$ be the $\rho$-neighborhood of $\Omega$. Suppose that $f$ is analytic in $\Omega_{\rho}$. Define the two-variable function $g(x,y)=f(x+\I y)$, and assume
    $$\norm{g}_{W(\Omega_{\rho})} \le B\quad \text{or}\quad \norm{g}_{H^2(\Omega_{\rho})} \le B.$$
    Then there exist $R>0$ and a smooth function $\hat{h}:D_R\to \CC$ such that
    \begin{equation}
        \norm{f(X+\I Y) - (2\pi)^{-1}\int_{u^2+v^2\le R^2} \hat{h}(u,v) e^{-\I(uX+vY)} \rd u \rd v}\le \epsilon 
    \end{equation}
    with the cost function
    $$\Lambda = \lambda_{\hat{h}} R = \mo{\frac{B}{\rho}\ln\frac{B}{\rho\epsilon}},$$
    where the constant depends only on the shape of $\Omega$.
\end{theorem}

\section{Weyl calculus for time-ordered evolution}\label{sec:time-ordered-weyl}
The two-dimensional Weyl calculus developed in \Cref{sec:weyl} applies directly to the function $f(z)=e^{-z}=e^{-(x+\I y)}$ on $\Omega=\{z:\Re(z)\ge0\}$. In \Cref{thm: f - h bound}, we may choose $\hat{h}=\sqrt{2\pi}\delta_1(v)r(u)\rd u$ as a product of two one-dimensional measures. Here $\delta_1(v)$ is the Dirac measure supported at $1$, and $r(u)\rd u$ is an absolutely continuous measure with inverse Fourier transform $\check{r}(x)\approx e^{-x}$ for $x\ge0$. However, \Cref{thm: f - h bound} does not directly apply to the time-dependent ODE
\begin{equation}\label{eq: time dependent ODE}
    \frac{\rd}{\rd t} \vec{\psi}(t) = -A(t)\vec{\psi}(t)
\end{equation}
which is an important application of LCHS. In this section, we show how to adapt the Weyl-calculus framework to the time-ordered operator
\begin{equation}\label{eq: time ordered operator}
    \mathcal{T}e^{-\int_0^t A(s) \rd s}
\end{equation}
that solves \eqref{eq: time dependent ODE}.

\begin{remark}
    The Weyl-calculus perspective is particularly useful for analyzing discrete LCHS formulas for time-ordered evolution. For continuous kernel functions, one may alternatively derive error bounds using contour-integral methods \cite{low2025optimal} or uniqueness of ODE solutions \cite{an2023linear}. However, for discrete LCHS formulas that do not necessarily arise from quadrature, the Weyl-calculus framework provides the most direct route to an error bound. We will illustrate this point more explicitly in \Cref{sec: exp-z opt}.
\end{remark}

Let $L(s)$ and $H(s)$ be Hermitian matrix-valued functions. Suppose a function $f$ admits the Fourier--Stieltjes representation
\begin{equation}\label{eq: 1d FS f}
    f(x) = (2 \pi)^{-1/2} \int_{\RR} e^{-\I ux} \rd \mu(u)
\end{equation}
for some finite complex Borel measure $\mu$. We define
\begin{equation}\label{eq: def whfl}
\wh{f}{L}:=(2 \pi)^{-1/2} \int_{\RR} \mathcal{T}e^{-\I\int_0^t (u L(s)+H(s))\rd s} \rd \mu(u).
\end{equation}

\begin{lemma}\label{lem: exp bound}
    For any $t \ge 0$, consider the Cartesian decomposition $A=L+\I H:[0,t]\to\mathbb{C}^{N\times N}$. Suppose $\|A\|_{L^1}<\infty$, and $l(s)I\preceq L(s)\preceq r(s)I$ for some scalar functions $l,r:[0,t]\to\RR$. If $f$ is a smooth function defined on the interval $\Omega=[\int_0^t l(s)\rd s,\int_0^t r(s)\rd s]$, then \begin{equation}\label{eq: whfl ext def}
    \wh{f}{L} := \wh{\tilde{f}}{L}, \quad \tilde{f}|_{\Omega} = f,\ \tilde{f}\in B(\RR)
    \end{equation} 
    is independent of the choice of extension $\tilde{f}$, and
    \begin{equation}\label{eq: wfl le fw}
        \norm{\wh{f}{L}}\le \norm{f}_{W(\Omega)}.
    \end{equation}
\end{lemma}
\begin{proof}
The proof parallels \Cref{thm: opf bound}. It has three parts, corresponding to the proofs of \Cref{lem: op C infty extension}, \Cref{lem: op determined by f Omega}, and \Cref{thm: opf bound}, respectively. The first two parts are for proving that $\wh{f}{L}$ is well-defined in \eqref{eq: whfl ext def}, and the third part is for the bound \eqref{eq: wfl le fw}.

\textbf{Part 1.} We first consider the case where $f\in\mathcal{S}(\RR)$ is a Schwartz function. This part follows the same strategy as \cite[Theorem 2.4]{jefferies2004spectral}. 
    Let $m(s) = \frac{l(s)+r(s)}{2}$. Then $\norm{L(s) - m(s)I}\le\frac{r(s)-l(s)}{2}$. For $z\in\CC$, consider the operator
    $$e(z) = \mathcal{T}e^{-\I\int_0^t (z (L(s)-m(s)I)+H(s))\rd s}$$
    Define the matrix-valued distribution
    $$T: \mathcal{S}(\RR) \to \matn: \varphi\mapsto (2\pi)^{-1/2}\int_{\RR} e(u) \hat{\varphi}(u) \rd u.$$
    By the definition of the Fourier transform of distributions,
    $$\langle T, \varphi\rangle = (2\pi)^{-1/2}\langle e, \hat{\varphi}\rangle = (2\pi)^{-1/2}\langle \hat{e}, \varphi\rangle,$$
    so $T = (2\pi)^{-1/2}\hat{e}$ as a distribution.
    
    For any matrix-valued function $B:[0,t]\to\mathbb{C}^{N\times N}$, the logarithmic-norm bound from~\cite{soderlind2006logarithmic} gives
\begin{align}
\left\|\mathcal{T}e^{\int_0^t B(s)\rd s}\right\|
\le \exp\left(\int_0^t\lambda_{\text{max}}\left(\frac{B(s)+B^\dagger(s)}{2}\right)\rd s\right).
\end{align}
Applying this bound to $B(s)=-\I(z(L(s)-m(s)I)+H(s))$ gives
\begin{align*}
    \norm{e(z)}&\le \exp\left(\int_0^t\lambda_{\text{max}}\left(\Im(z)(L(s) - m(s)I)\right)\rd s\right)\\
    &\le \exp\left(\abs{\Im(z)}\int_0^t\norm{L(s) - m(s)I}\rd s\right)\\
    &\le \exp\left(\abs{\Im(z)}\int_0^t\frac{r(s)-l(s)}{2}\rd s\right).
\end{align*}
By \cite[Lemma 3]{low2025optimal}, $e(z)$ is entire. The Paley--Wiener theorem \cite[Proposition 2.1]{jefferies2004spectral} then implies that $T=(2\pi)^{-1/2}\hat{e}$ is a matrix-valued tempered distribution with support contained in $[-\int_0^t\frac{r(s)-l(s)}{2}\rd s, \int_0^t\frac{r(s)-l(s)}{2}\rd s]$.

By the definition of $e(u)$,
$$\wh{f}{L}:=(2 \pi)^{-1/2} \int_{\RR} e^{-\I u\int_0^t m(s)\rd s} e(u) \hat{f}(u)\rd u = \langle T, f(\, \cdot\, - \textstyle\int_0^t m(s)\rd s)\rangle,$$
which shows that $\wh{\cdot}{L}$ is a distribution supported on the interval $\Omega=[\int_0^t l(s)\rd s,\int_0^t r(s)\rd s]$.

\textbf{Part 2.} Part 1 shows that if $f\in\mathcal{S}(\RR)$ vanishes in a neighborhood of $\Omega$, then $\wh{f}{L}=0$. We now relax this to functions $f$ satisfying \eqref{eq: 1d FS f} and vanishing on $\Omega$. The argument is the same as in the proof of \Cref{lem: op determined by f Omega}, so we omit the details.

\textbf{Part 3.} The proof of \eqref{eq: wfl le fw} is similar to \Cref{thm: opf bound}. For any function $\tilde{f}$ that extends $f$ to $\RR$ and satisfies
    $$\tilde{f}(x) = (2 \pi)^{-1/2} \int_{\RR} e^{-\I ux} \rd \tilde{\mu}(u)$$
    for some finite complex Borel measure $\tilde{\mu}$, Part 2 gives $\wh{f}{L} = \wh{\tilde{f}}{L}$. By definition,
    $$ \norm{\wh{\tilde{f}}{L}}=\norm{(2 \pi)^{-1/2} \int_{\RR} \mathcal{T}e^{-\I\int_0^t (u L(s)+H(s))\rd s} \rd \tilde{\mu}(u)}\le (2 \pi)^{-1/2} \int_{\RR}  \rd \abs{\tilde{\mu}}\le \norm{\tilde{f}}_{W}.$$ 
    Taking the infimum over all extensions $\tilde{f}$ completes the proof by the definition of $\norm{f}_{W(\Omega)}$.
\end{proof}

\begin{theorem}\label{cor: exp}
    For any $t \ge 0$, consider the Cartesian decomposition $A=L+\I H:[0,t]\to\mathbb{C}^{N\times N}$. Suppose $\|A\|_{L^1}<\infty$ and $L(s)\succeq0$. If $f$ is a smooth function defined on $[0,+\infty)$, then
    \begin{equation}\label{eq: diff exp-x}
        \norm{\wh{f}{L} - \mathcal{T}e^{-\int_0^t A(s) \rd s}}\le \norm{f-e^{-x}}_{W([0,+\infty))}.
    \end{equation}
\end{theorem}
\begin{proof}
    Let $g(x)=e^{-|x|}$. From \cite{an2023linear}, we have
    $$\wh{g}{L} = \frac{1}{\pi} \int_{\RR} \mathcal{T}e^{-\I\int_0^t (u L(s)+H(s))\rd s} \frac{1}{1+u^2}\rd u= \mathcal{T}e^{-\int_0^t A(s) \rd s}.$$
    Suppose $\int_0^t \norm{L(s)} \rd s = B$. Applying \Cref{lem: exp bound} on the interval $\Omega=[0,B]$ gives
    \begin{equation*}
        \norm{\wh{f}{L} - \mathcal{T}e^{-\int_0^t A(s) \rd s}} = \norm{\wh{f-g}{L}}\le\norm{f-g}_{W([0,B])}\le \norm{f-e^{-x}}_{W([0,+\infty))}.
    \end{equation*}
\end{proof}

\subsection{Improved prefactor for dissipative dynamics simulation}\label{sec:improved-prefactor}

In this subsection, we present a new kernel function for the LCHS simulation of
$$\mathcal{T}e^{-\int_0^t A(s) \rd s}$$
for a time-dependent matrix $A(t)$. This is the first and most extensively studied application of LCHS. Ref.~\cite{low2025optimal} gives an efficient kernel construction using a contour-integral argument; here we construct a kernel with a better prefactor.

For positive parameters $R$ and $C$, define
\begin{align*}
    \hat{g}_{R,C}(u) = \sqrt{\frac{2}{\pi}} e^{2C \arctan\left(\frac{1}{R}\right)} e^{\I C \log\left(\frac{R-u}{R+u}\right)} \frac{1}{1+u^2}
\end{align*}
on the interval $(-R,R)$, and set
\begin{align}\label{eq: fRC definition}
    f_{R,C}(x) := \frac{1}{\sqrt{2\pi}}\int_{-R}^R\hat{g}_{R,C}(u) e^{-\I ux} \rd u.
\end{align}
\begin{lemma}\label{lem: fRC approximation}
The following bound holds:
    $$\norm{f_{R,C}-e^{-x}}_{W([0,+\infty))} \le \frac{2}{\pi} e^{-2 C \arctan(R)}\arctan\left(\frac{1}{R}\right). $$
\end{lemma}
\begin{proof}
    The function $\hat{g}_{R,C}(z)$ can be analytically continued to the lower half-plane 
    $\{z\in \CC: \Im(z)\le 0\}\setminus\{\pm R\}$, with branch points at $z = \pm R$. Define
    \begin{align}\label{eq: int grc}
        g_{R,C}(x) = \left( \int_{-R}^R + \int_R^\infty + \int_{-\infty}^{-R} \right) \frac{1}{\sqrt{2\pi}} \hat{g}_{R,C}(u)  e^{-\I ux} \rd u,
    \end{align}
    We first prove that $g_{R,C}(x)=e^{-x}$ for $x\ge0$.
    
    We evaluate the integral in \eqref{eq: int grc} using a closed contour $\mathcal{C}$ that encloses a large semi-disk of radius $K$ in the lower half-plane, with small semicircular indentations of radius $\delta$ around $\pm R$, as shown in \Cref{fig:contour}. 

    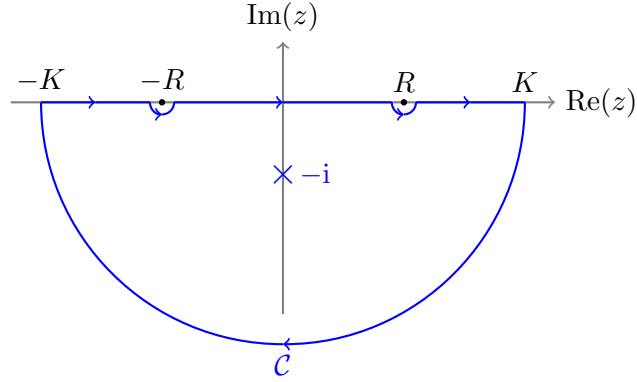
\begin{figure}[htbp]
    \centering
    \begin{tikzpicture}[scale=0.8, thick]
        \draw[->, gray] (-4.5, 0) -- (4.5, 0) node[right, black] {$\Re(z)$};
        \draw[->, gray] (0, -3.5) -- (0, 1.0) node[above, black] {$\Im(z)$};
        
        \node[blue] at (0, -1.2) {\Large $\times$};
        \node[blue, right] at (0.1, -1.2) {$-\I$};
        
        \draw[->, blue, thick] (-1.8, 0) -- (0, 0);
        \draw[blue, thick] (0, 0) -- (1.8, 0);
        
        \draw[->, blue, thick] (1.8, 0) arc (180:270:0.2);
        \draw[blue, thick] (2.0, -0.2) arc (270:360:0.2);
        
        \draw[->, blue, thick] (2.2, 0) -- (3.1, 0);
        \draw[blue, thick] (3.1, 0) -- (4.0, 0);
        
        \draw[->, blue, thick] (4.0, 0) arc (360:270:4.0) node[below] {$\mathcal{C}$};
        \draw[blue, thick] (0, -4.0) arc (270:180:4.0);
        
        \draw[->, blue, thick] (-4.0, 0) -- (-3.1, 0);
        \draw[blue, thick] (-3.1, 0) -- (-2.2, 0);
        
        \draw[->, blue, thick] (-2.2, 0) arc (180:270:0.2);
        \draw[blue, thick] (-2.0, -0.2) arc (270:360:0.2);
        
        \filldraw[black] (2,0) circle (1pt) node[above] {$R$};
        \filldraw[black] (-2,0) circle (1pt) node[above] {$-R$};
        \node[above] at (4.0, 0) {$K$};
        \node[above] at (-4.0, 0) {$-K$};
    \end{tikzpicture}
    \caption{Integration contour $\mathcal{C}$ in the lower half-plane, with small indentations around the branch points $\pm R$ and a single enclosed pole at $-\I$.}
    \label{fig:contour}
\end{figure}

    Inside $\mathcal{C}$, the analytically continued integrand $\hat{g}_{R,C}(z)e^{-\I zx}$ has a single simple pole at $z=-\I$, coming from the factor $(1+z^2)^{-1}$. Its residue gives
    \begin{align*}
        -2\pi \I \, \Res_{z=-\I} \left[ \frac{1}{\sqrt{2\pi}} \hat{g}_{R,C}(z) e^{-\I zx} \right] = e^{-x}.
    \end{align*}
    The contour $\mathcal{C}$ is traversed clockwise. By Cauchy's residue theorem, the integral over $\mathcal{C}$ is exactly the pole contribution $e^{-x}$.

    Under the continuation to $\Im(z)<0$, we have $\Arg\left(\frac{R-z}{R+z}\right)\in(0,\pi)$. Hence $\abs{e^{\I C \log\left(\frac{R-z}{R+z}\right)}}\le1$, so $\hat{g}_{R,C}(z)$ remains bounded near the branch points. Consequently, the integrals over the small semicircles around $\pm R$ vanish as $\delta\to0$. Along the large semicircle of radius $K$, parameterized by $z=Ke^{\I\theta}$ for $\theta\in[-\pi,0]$, the contribution is bounded by
    \begin{align*}
        \int_{-\pi}^{0} K \abs{\frac{1}{1+K^2 e^{2\I\theta}}} \rd \theta \xrightarrow{K\to \infty} 0.
    \end{align*}
    Therefore, the integrals over the real segments converge to the residue value:
    \begin{align*}
        g_{R,C}(x) = \left( \int_{-R}^R + \int_R^\infty + \int_{-\infty}^{-R} \right) \frac{1}{\sqrt{2\pi}} \hat{g}_{R,C}(u) e^{-\I ux} \mathrm{d}u = e^{-x},
    \end{align*}
    where the values for $|u|>R$ are taken from the analytic continuation through the lower half-plane. It follows that
    \begin{align*}
        \norm{f_{R,C}-e^{-x}}_{W([0,+\infty))} = \norm{f_{R,C}-g_{R,C}}_{W([0,+\infty))} \le \frac{1}{\sqrt{2\pi}} \left( \int_R^\infty + \int_{-\infty}^{-R} \right) \abs{\hat{g}_{R,C}(u)} \rd u.
    \end{align*}
    For $u>R$, analytic continuation of the logarithm through the lower half-plane adds $\I\pi$ to the phase:
    $\log\left(\frac{R-u}{R+u}\right)=\ln\left|\frac{u-R}{u+R}\right|+\I\pi$. This phase shift gives the uniform exponential suppression
    \begin{align*}
        \abs{\exp\left[ \I C \left(\log\left| \frac{u-R}{u+R} \right| + \I\pi\right) \right]} = e^{-C\pi}.
    \end{align*}
    By symmetry, the interval $(-\infty,-R)$ gives the same suppression factor. Thus
    \begin{align*}
        \norm{f_{R,C}-e^{-x}}_{W([0,+\infty))} &\le \frac{1}{\pi} e^{2C \arctan(1/R)} e^{-C\pi} \left( \int_R^\infty \frac{\mathrm{d}u}{1+u^2} + \int_{-\infty}^{-R} \frac{\mathrm{d}u}{1+u^2} \right) \\
        &= \frac{2}{\pi} e^{-C\pi + 2C \arctan(1/R)} \arctan\left(\frac{1}{R}\right).
    \end{align*}
    Using the identity $\frac{\pi}{2}-\arctan(1/R)=\arctan(R)$, the exponent simplifies to $-2C\arctan(R)$, completing the proof.
\end{proof}

Combining the preceding lemma with \Cref{cor: exp} and optimizing the parameters $R$ and $C$ gives the following theorem.

\begin{theorem}\label{thm: exp approx 1}
    Assume $\epsilon<e^{-e}$. For any $t\ge0$, consider the Cartesian decomposition $A=L+\I H:[0,t]\to\mathbb{C}^{N\times N}$. Suppose $\|A\|_{L^1}<\infty$ and $L(s)\succeq0$. Then
    \begin{equation}\label{eq: exp approx 1}
        \norm{\mathcal{T}e^{-\int_0^t A(s) \rd s} - \frac{1}{\sqrt{2\pi}} \int_{-R}^R \hat{g}_{R,C}(u) \mathcal{T}e^{-\I\int_0^t (u L(s)+H(s))\rd s} \rd u}\le \epsilon
    \end{equation}
    holds for $C = \frac{1}{\pi} \ln(\epsilon^{-1})$ and $R = \tan\left(\frac{\pi}{2}\left(1 - \frac{1}{\ln(\epsilon^{-1})}\right)\right)$. The cost function $\Lambda=\lambda_{\hat{g}_{R,C}}R$ satisfies
    $$\Lambda \le \frac{2 e}{\pi}\left(\ln\frac{1}{\epsilon} - 1\right).$$
\end{theorem}
\begin{proof}
    By \Cref{lem: fRC approximation} and \Cref{cor: exp}, we have the error bound
\begin{equation}\label{eq:thm_general_error_bound}
        \norm{\mathcal{T}e^{-\int_0^t A(s) \rd s} - \frac{1}{\sqrt{2\pi}} \int_{-R}^R \hat{g}_{R,C}(u) \mathcal{T}e^{-\I\int_0^t (u L(s)+H(s))\rd s} \rd u}\le \frac{2}{\pi} e^{-2 C \arctan(R)}\arctan\left(\frac{1}{R}\right).
    \end{equation}
The normalization factor is
$$\lambda_{\hat{g}_{R,C}} = \frac{1}{\sqrt{2\pi}}\int_{-R}^R|\hat{g}_{R,C}(k)|\mathrm{d}k = \frac{2}{\pi} e^{2C \arctan(\frac{1}{R})} \arctan(R).$$
To minimize $\Lambda(R,C)=\lambda_{\hat{g}_{R,C}}R$ subject to the error bound in \eqref{eq:thm_general_error_bound}, consider the constrained optimization problem
\begin{equation}\label{eq: opt prob}
\begin{aligned}
    \min_{R>0, C>0} \quad & \Lambda(R, C) =  \frac{2}{\pi} R e^{2C \arctan(\frac{1}{R})} \arctan(R)  \\
    \text{s.t.} \quad & \frac{2}{\pi} e^{-2 C \arctan(R)}\arctan\left(\frac{1}{R}\right) \le \epsilon.
\end{aligned}
\end{equation}
    Choose $C$ and $R$ such that $2C = \frac{1}{\arctan(1/R)} = \frac{2}{\pi}\ln\frac{1}{\epsilon}$. Then $2C \arctan(\frac{1}{R}) = 1$ and $2C \arctan(R) = \pi C - 1 = \ln\frac{1}{\epsilon} - 1.$ The constraint is satisfied because
    $$\frac{2}{\pi} e^{-2 C \arctan(R)}\arctan(\frac{1}{R}) = \epsilon \frac{e}{\ln\frac{1}{\epsilon}}\le\epsilon,$$
    where the last step uses $\epsilon<e^{-e}$.

    We also have $R \le \frac{1}{\arctan(1/R)} = \frac{2}{\pi}\ln\frac{1}{\epsilon}$, so the cost is at most
    $$\frac{2}{\pi} R e^{2C \arctan(\frac{1}{R})} \arctan(R) = \frac{2R}{\pi}  e \left(\frac{\pi}{2} - \frac{\pi}{2}\frac{1}{\ln\frac{1}{\epsilon}}\right) \le \frac{2 e}{\pi}\left(\ln\frac{1}{\epsilon} - 1\right).$$
\end{proof}

For comparison, the state-of-the-art scaling from Ref.~\cite{low2025optimal} is asymptotically $2.57\log_2(\frac{1}{\epsilon}) \approx 3.71\ln\frac{1}{\epsilon}$. Our explicit parameter choice in the new kernel yields $\frac{2e}{\pi}\ln\frac{1}{\epsilon} \approx 1.73\ln\frac{1}{\epsilon}$, giving a $2.1\times$ improvement in the asymptotic constant. 

An even tighter cost estimate is obtained by numerically optimizing \eqref{eq: opt prob}; the results are shown in \Cref{tab:numerical_opt}.

\begin{table}[htbp]
\centering
\begin{tabular}{cccccc}
\toprule
\multirow{2}{*}{$\epsilon$} & \multirow{2}{*}{$R$} & \multirow{2}{*}{$C$} & \multirow{2}{*}{$\lambda_{\hat{g}_{R,C}}$} & \multicolumn{2}{c}{Cost $\Lambda(R, C) = \lambda R$} \\
\cmidrule(lr){5-6}
& & & & \textbf{Our Work} & Prior Art \cite{low2025optimal} \\
\midrule
$10^{-1}$ & 1.49 & 0.68 & 1.38 & \textbf{2.07} & 2.55 \\
$10^{-2}$ & 2.83 & 1.25 & 1.83 & \textbf{5.18} & 7.06 \\
$10^{-3}$ & 4.17 & 1.88 & 2.05 & \textbf{8.57} & 12.74 \\
$10^{-4}$ & 5.53 & 2.53 & 2.19 & \textbf{12.11} & 19.26 \\
$10^{-5}$ & 6.90 & 3.20 & 2.28 & \textbf{15.73} & 26.42 \\
$10^{-6}$ & 8.28 & 3.88 & 2.35 & \textbf{19.41} & 34.08 \\
$10^{-7}$ & 9.66 & 4.56 & 2.39 & \textbf{23.14} & 42.15 \\
$10^{-8}$ & 11.06 & 5.26 & 2.43 & \textbf{26.90} & 50.56 \\
$10^{-9}$ & 12.46 & 5.95 & 2.46 & \textbf{30.68} & 59.27 \\
$10^{-10}$ & 13.87 & 6.65 & 2.49 & \textbf{34.48} & 68.23 \\
\bottomrule
\end{tabular}
\caption{Numerical verification of the optimized scaling parameters. For each target error $\epsilon$, we report the optimized radius $R$, phase parameter $C$, normalization factor $\lambda_{\hat{g}_{R,C}}$, and the corresponding LCHS cost prefactor $\lambda_{\hat{g}_{R,C}} R$. The results are compared with the best prior-art bounds from Ref.~\cite{low2025optimal}, demonstrating a significant constant-factor improvement.}
\label{tab:numerical_opt}
\end{table}

\section{Cost optimization of discrete LCHS formulas}\label{sec:convex-optimization}
The continuous LCHS formula \eqref{eq: lchs formula} is not directly implementable in a coherent quantum circuit \Cref{sec: coherent}. Previous works \cite{an2023linear, an2026quantum, an2026laplace, low2025optimal, huang2025fourier} discretize such formulas using quadrature rules, and the same approach applies to the explicit continuous constructions in this paper. However, quadrature-based discretization can be inefficient in two ways:
\begin{enumerate}
    \item If the weight function $\hat{h}$ is not sufficiently smooth, or if the quadrature rule is not spectrally convergent, the number of discretization terms can be $\mo{\poly(\epsilon^{-1})}$. In the LCU implementation, this leads to an additional $\mo{\polylog(\epsilon^{-1})}$ ancilla-qubit cost in the $\mathrm{PREP}$ operators \eqref{eq: prep operators}.
    \item The cost $\Lambda$ associated with an explicit construction $\hat{h}(u,v)$ may have optimal asymptotic scaling but suboptimal constants. A direct numerical optimization method can give more favorable resource estimates.
\end{enumerate}

For the exponential function $f(z)=e^{-z}$, the works \cite{low2025optimal, huang2025fourier} address the second issue by optimizing over parametrized families of weight functions. The output, however, is still a continuous kernel and must be discretized afterward. Moreover, the parametrizations are chosen from intuition, which may not be optimal and may be difficult to generalize beyond $e^{-z}$.

The Weyl-calculus framework developed in this paper suggests a direct alternative: optimize the weights of the discrete LCHS formula themselves, without first choosing a continuous ansatz.

Assume that the truncation radius $R$ and total weight budget $\lambda$ are fixed, and choose a grid
$$\mathcal{G} = \{(u_j, v_j)\}_{j=1,\ldots,J}\subset D_R.$$
We seek weights $\bw=(w_j)_{1\le j\le J}$ solving the convex optimization problem
\begin{equation}\label{eq: general opt problem}
\begin{aligned}
    \epsilon = \min_{\bw}  &\ \norm{f(x+\I y) - \sum_{j=1}^J w_j\, e^{-\I(u_jx+v_jy)}}_{W(\Omega)}\\
    &\text{s.t.}\quad \sum_{j=1}^J|w_j|\le \lambda
\end{aligned}
\end{equation}
where $\Omega$ is a region containing the numerical range of $X+\I Y$. The function $h(x,y)=\sum_{j=1}^J w_j e^{-\I(u_jx+v_jy)}$ is the inverse Fourier transform of the measure $\mu=2\pi\sum_{j=1}^J w_j\delta_{(u_j,v_j)}$. By \Cref{thm: f - h bound},
\begin{equation}
    \norm{f(X+\I Y) - \sum_{j=1}^J w_j\, e^{-\I(u_jX+v_jY)}}\le \epsilon,
\end{equation}
and the cost is $\Lambda=\lambda R$. In practice, one may instead fix $\Lambda$ and optimize the resulting error $\epsilon$. This can be done by trying different values of $R$, setting $\lambda=\Lambda/R$, and choosing the smallest certified error.

Computing the exact local Wiener norm is itself difficult. Therefore, when implementing \eqref{eq: general opt problem}, we replace it by explicit upper bounds. For common domains such as intervals, squares, and discs, the resulting convex programs are described below.

\subsection{Convex optimization on $[0,T]$: application to dissipative dynamics}\label{sec: exp-z opt}
To obtain a discrete LCHS formula for $\mathcal{T}e^{-\int_0^t A(s)\rd s}$ under the dissipativity condition $L(s)=\frac{A(s)+A(s)^\dagger}{2}\succeq0$, \Cref{lem: exp bound} suggests choosing $T\ge\norm{L}_{L^1([0,t])}=\int_0^t\norm{L(s)}\rd s$ and approximating $e^{-x}$ on $[0,T]$. The same lemma also applies to non-dissipative dynamics by approximating $e^{-x}$ on the interval $[\int_0^t l(s)\rd s,\int_0^t r(s)\rd s]$; this extension is straightforward, so we omit it for conciseness.

Let $s=\frac{\pi}{T}$, and choose a truncation interval $[-R,R]$ such that $N=R/s$ is an integer. We use the finite Fourier expansion
\begin{equation}
    S(x)=s\sum_{j=-N}^{N}g_j e^{-\I u_jx},
    \qquad u_j=js,
\end{equation}
to approximate $e^{-x}$ for $x\in[0,T]$, where the coefficients $g_j\in\CC$ are variables to be determined. The objective is to minimize the local Wiener norm of the residual
\[
    \eta(x)=S(x)-e^{-x},\qquad x\in[0,T].
\]
Denote the cosine expansion of $\eta$ on $[0,T]$ by
\[
    \eta(x)=\frac{a_0}{2}+\sum_{m=1}^{\infty}a_m\cos(c_mx),
    \qquad c_m=\frac{m\pi}{T}.
\]
Then
\[
    \norm{\eta}_{W([0,T])}\le \frac{|a_0|}{2}+\sum_{m=1}^{\infty}|a_m|.
\]
We use the cosine expansion coefficients to upper bound the local Wiener norm because they are absolutely summable for a smooth function on a finite interval. In particular, for a fixed smooth residual $\eta$, the coefficients $a_m$ decay at least quadratically in $m$.

The coefficients $a_m$ can be computed analytically. Define
\begin{align}
    D_{j,m}
    &=\frac{2}{T}\int_0^T e^{-\I u_jx}\cos(c_mx)\rd x \notag\\
    &=\frac{1}{T}\left(
        \frac{1-e^{-\I(u_j-c_m)T}}{\I(u_j-c_m)}
        +\frac{1-e^{-\I(u_j+c_m)T}}{\I(u_j+c_m)}
    \right),\label{eq: Djm exp showcase}
\end{align}
where the values at removable singularities are clear. Also set
\[
    C_m=\frac{2}{T}\int_0^T e^{-x}\cos(c_mx)\rd x
    =\frac{2/T}{1+c_m^2}\left[1-(-1)^m e^{-T}\right].
\]
Thus
\[
    a_m=s\sum_{j=-N}^{N}g_jD_{j,m}-C_m .
\]

For fixed $R$ and $\lambda$, truncating the cosine series at a sufficiently large value $M$ gives the finite convex program
\begin{equation}\label{eq: exp cvx constraint}
    \epsilon_M=
    \min_{\{g_j\}_{j=-N}^{N}}\ 
    \frac{1}{2}\left|s\sum_{j=-N}^{N}g_jD_{j,0}-C_0\right|
    +\sum_{m=1}^{M}\left|s\sum_{j=-N}^{N}g_jD_{j,m}-C_m\right|
\end{equation}
subject to
\[
    s\sum_{j=-N}^{N}|g_j|\le \lambda .
\]
This is a second-order cone program (SOCP), since both the objective and the constraint are sums of moduli of complex affine functions. It can be solved efficiently using standard convex optimization solvers. To verify a solution, we use a larger value $M_{\text{test}}$ in the right-hand side of \eqref{eq: exp cvx constraint}, giving a sharper posterior error bound.

The practical implementation imposes the conjugate symmetry $\overline{g_{-j}}=g_j$, since the target function $e^{-x}$ is real-valued. This reduces the number of variables by roughly a factor of two.

\begin{figure}[!htbp]
    \centering
    \subfigure[Cost comparison]{
        \includegraphics[width=0.45\textwidth]{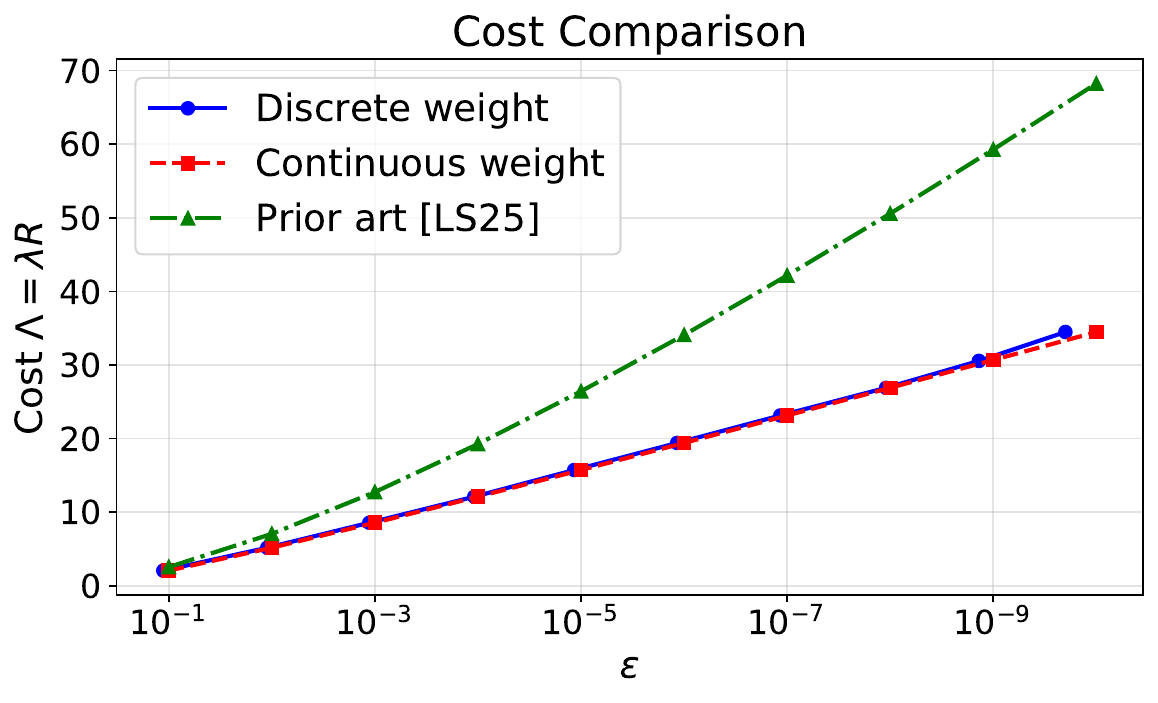}
        \label{fig:cost_comparison}
    }
    \hfill
    \subfigure[Optimized weights and reference kernel]{
        \includegraphics[width=0.45\textwidth]{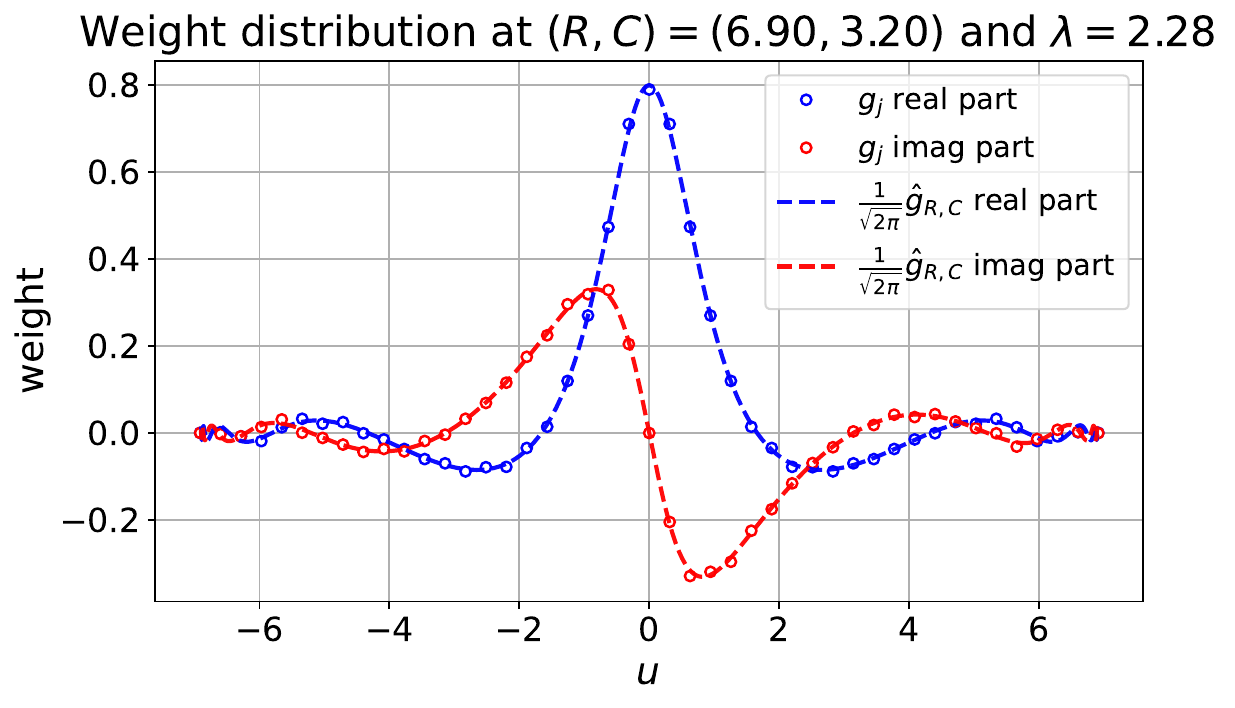}
        \label{fig:kernel_theory_reference}
    }
    \caption{Numerical illustration of the optimized LCHS construction for $e^{-x}$. (a) the certified cost-error tradeoff for the optimized discrete weights, compared with the continuous kernel and the prior-art bound. (b) optimized discrete weights for $T=10$, $R=6.90$, $C=3.20$, and $\lambda=2.28$, compared with the reference values from $\frac{1}{\sqrt{2\pi}}\hat{g}_{R,C}$.}
    \label{fig:exp_numerical_illustration}
\end{figure}

We run this solver using the $R$ and $\lambda$ values from \Cref{tab:numerical_opt}. The resulting $\epsilon$ values, shown in \Cref{fig:cost_comparison}, are almost identical to those obtained from the continuous distribution $\hat{g}_{R,C}$, confirming the effectiveness of the discrete optimization. These results use $T=50$, and we also verify that they are insensitive to this choice. We use $M_{\text{opt}}=3000$ in the optimization and certify the residual error with \eqref{eq: exp cvx constraint} using $M_{\text{test}}=20000$.

\Cref{fig:kernel_theory_reference} shows the optimized discrete weights $g_j$ for $T=10$, $R=6.90$, and $\lambda=2.28$. The weights are close to the values $\frac{1}{\sqrt{2\pi}}\hat{g}_{R,C}(u_j)$. In fact, the closed-form kernel $\hat{g}_{R,C}$ was originally inspired by this optimization result. Even with the explicit formula for $\hat{g}_{R,C}$, the optimized discrete weights remain useful for implementation: they live on a grid of size only $\mo{R}=\mo{\log\frac{1}{\epsilon}}$. By comparison, discretizing the continuous integral \eqref{eq: fRC definition} by quadrature may require $\mo{\poly(\frac{1}{\epsilon})}$ terms because $\hat{g}_{R,C}$ has singularities near $\pm R$.

\subsection{Convex optimization on $[-1,1]^2$: application to matrix square roots}\label{sec: square opt}

If the numerical range of $A$ is contained in the square $\Omega = [-1,1]^2$, we can approximate $f(x+\I y)$ using
\begin{equation}
    S(x,y) = \sum_{(j,k)\in\mathcal{G}_R} w_{j,k} e^{-\I\frac{\pi}{2}(jx+ky)},
\end{equation}
on $\Omega$, where
\begin{equation}\label{eq: disc frequency set}
    \mathcal{G}_R
    =\left\{(j,k)\in\ZZ^2:
    \left(\frac{\pi}{2}\right)^2(j^2+k^2)\le R^2\right\}.
\end{equation}
The optimization problem is similar to the one-dimensional case, and the local Wiener norm can again be upper-bounded using a cosine expansion. Suppose that the residual $\eta(x,y)=S(x,y)-f(x+\I y)$ has the cosine expansion
\begin{equation}
    \eta(x,y) = \sum_{m,n=0}^\infty a_{m,n} \cos\left(\frac{m\pi}{2}(x+1)\right) \cos\left(\frac{n\pi}{2}(y+1)\right).
\end{equation}
Then
\begin{equation}
    \norm{\eta}_{W(\Omega)} \le \sum_{m,n=0}^\infty |a_{m,n}|.
\end{equation}
The coefficients $a_{m,n}$ can be computed analytically by taking inner products with the basis functions. In particular, set
\begin{equation}\begin{aligned}
    D_{m,n}^{j,k} &= \frac{1}{(1+\delta_{m,0})(1+\delta_{n,0})}\int_{-1}^1\int_{-1}^1 e^{-\I\frac{\pi}{2}(jx+ky)} \cos\left(\frac{m\pi}{2}(x+1)\right) \cos\left(\frac{n\pi}{2}(y+1)\right) \mathrm{d}x \mathrm{d}y, \\
    C_{m,n} &= \frac{1}{(1+\delta_{m,0})(1+\delta_{n,0})}\int_{-1}^1\int_{-1}^1 f(x+\I y) \cos\left(\frac{m\pi}{2}(x+1)\right) \cos\left(\frac{n\pi}{2}(y+1)\right) \mathrm{d}x \mathrm{d}y.
\end{aligned}\end{equation}
The resulting optimization problem is
\begin{equation}\begin{aligned}
    \epsilon_{M} = \min_{\{w_{j,k}\}} &\ \sum_{m=0}^M \sum_{n=0}^M \left| \sum_{(j,k)\in\mathcal{G}_R} w_{j,k} D_{m,n}^{j,k} - C_{m,n} \right| \\
    &\text{s.t.}\quad \sum_{(j,k)\in\mathcal{G}_R} |w_{j,k}| \le \lambda,
\end{aligned}\end{equation}
As in the one-dimensional case, we use a sufficiently large $M_{\text{opt}}$ during optimization and a larger $M_{\text{test}}$ for posterior verification.

We apply this optimization approach to compute LCHS weights for the matrix square root $\sqrt{A}$, assuming that the numerical range of $A$ is contained in the open right half-plane $\{\Re(z)>0\}$. To adapt the problem to the domain $\Omega=[-1,1]^2$, we apply a linear transformation $A'=\beta A-(1+\rho)I$ so that the numerical range of $A'$ is contained in $\Omega$. Thus we focus on functions
$$f(z) = \sqrt{(1+\rho)+z}$$
with $\rho>0$, and approximate them on $\Omega=[-1,1]^2$.

We present numerical results for $\rho=1,0.5,0.25,0.125$. For a fixed budget $\Lambda$, we try different values of $R$, set $\lambda=\Lambda/R$, run the optimization with $M_{\text{opt}}=60$ and $M_{\text{test}}=1000$, and report the smallest certified $\epsilon$. The results are shown in \Cref{fig:sqrt_cost_comparison}.

It is useful to compare these numerical results with the theoretical bounds. The bound in \Cref{thm: construction 1 for disc} applies because $D_1\subset[-1,1]^2$, so controlling the square residual also controls the disc residual. The parameter $\rho$ has the same meaning here because $\sqrt{(1+\rho)+z}$ has finite local Wiener norm on $D_{1+\rho}$. In fact, $\sqrt{(1+\rho)+z}$ can be extended to an $H^{1.25}(\RR^2)$ function, and one can prove $\norm{g}_{W(\RR^2)}\le C\norm{g}_{H^{1.25}(\RR^2)}$ for some constant $C$ by replacing the $(1+u^2+v^2)^2$ factor in \eqref{eq: H2 estimate} with $(1+u^2+v^2)^{1.25}$. A simple lower bound is $\norm{\sqrt{(1+\rho)+z}}_{W(D_{1+\rho})}\ge\norm{\sqrt{(1+\rho)+z}}_{L^{\infty}(D_{1+\rho})}\ge\sqrt{2}$. Therefore, the results in \Cref{fig:sqrt_cost_comparison} have lower cost than even the most optimistic upper bound obtainable from \Cref{thm: construction 1 for disc}, namely $\Lambda \sim \frac{3\sqrt{2}}{\rho}\ln \frac{30\sqrt{2}}{\rho\epsilon}$.

\begin{figure}[htbp]
    \centering
    \includegraphics[width = 0.8\textwidth]{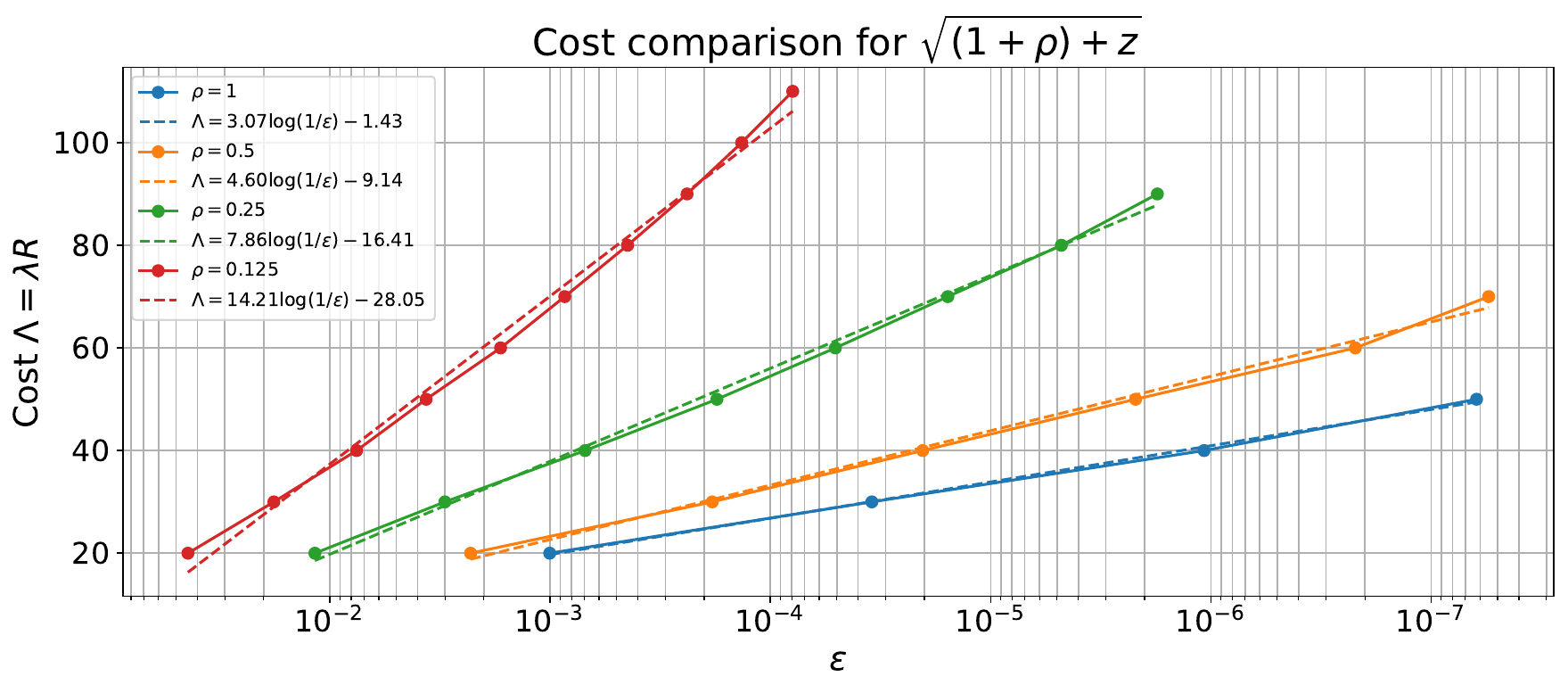}
    \caption{Cost comparison for $f(z)=\sqrt{(1+\rho)+z}$ for several values of $\rho$. The dashed lines are linear fits of $\Lambda$ versus $\log(1/\epsilon)$.}
    \label{fig:sqrt_cost_comparison}
\end{figure}

\subsection{Convex optimization on $D_1$: example of $f(z)=z^4$}

If a function has moderate magnitude on the unit disc $D_1$ but large values on the rest of the square $[-1,1]^2$, it is wasteful to measure the residual by a tensor-product cosine expansion on the whole square. In this case, we instead optimize the residual directly on $D_1$.
We still approximate $f(x+\I y)$ on $D_1$ by
\begin{equation}\label{eq: disc trig approximation}
    S(x,y) = \sum_{(j,k)\in\mathcal{G}_R} w_{j,k}
    e^{-\I\frac{\pi}{2}(jx+ky)},
\end{equation}
where $w_{j,k}\in\CC$. 

We use the Neumann Fourier--Bessel basis on the disc to measure the residual. Write $x=r\cos\theta$ and $y=r\sin\theta$. Let $J_p$ be the Bessel function of the first kind, and let $\beta_{p,\ell}$ be the $\ell$-th non-negative zero of $J_p'$. Thus $\beta_{0,1}=0$, while otherwise $\beta_{p,\ell}>0$. This gives the following unnormalized orthogonal basis for $L^2(D_1)$:
\begin{equation}\label{eq: disc FB basis}
    \Psi_{m,\ell}(r,\theta)
    =J_{|m|}(\beta_{|m|,\ell}r)e^{\I m\theta},
    \qquad m\in\ZZ,\quad \ell=1,2,\ldots .
\end{equation}
The mode $(m,\ell)=(0,1)$ is the constant function. The zero convention, the Neumann boundary condition, and the orthogonality relation are recalled in \Cref{appsubsec:disc-fb-modes}. The normalization constant is
\begin{equation}\label{eq: disc FB norm constant}
    (2\pi)N_{p,\ell}=\norm{\Psi_{m,\ell}}_{L^2(D_1)}^2 = (2\pi)\int_0^1J_p(\beta_{p,\ell}r)^2r\rd r.
\end{equation}
For the residual
\[
    \eta(r,\theta)=S(r\cos\theta,r\sin\theta)-f(re^{\I\theta}),
\]
we expand it as
\begin{equation}\label{eq: bessel a_ml}
    \eta=\sum_{m\in\ZZ}\sum_{\ell\ge1}a_{m,\ell}\Psi_{m,\ell}.
\end{equation}
In \Cref{appsubsec:disc-fb-wiener-bound}, we establish $\norm{\Psi_{m,\ell}}_{W(D_1)}\le 1$.
Consequently, 
\[
    \norm{\eta}_{W(D_1)}
    \le \sum_{m\in\ZZ}\sum_{\ell\ge1}|a_{m,\ell}|.
\]
This is the disc analogue of the cosine coefficient bound used in the one-dimensional and square cases. 

It remains to write the residual coefficients as affine functions of the unknown weights. The Neumann Fourier--Bessel expansion of the plane wave $e^{-\I\frac{\pi}{2}(jx+ky)}$ is
$$e^{-\I\frac{\pi}{2}(jx+ky)} = \sum_{m\in\ZZ}\sum_{\ell\ge1}P_{m,\ell}^{j,k}\Psi_{m,\ell}(r,\theta),$$
where the coefficient $P_{m,\ell}^{j,k}$ has a closed-form expression recorded in \Cref{appsubsec:disc-fb-matrix}.
The target coefficients are
\begin{equation}\label{eq: disc target coefficient}
    C_{m,\ell}
    =\frac{1}{2\pi N_{|m|,\ell}}
    \int_0^{2\pi}\int_0^1
    f(re^{\I\theta})J_{|m|}(\beta_{|m|,\ell}r)e^{-\I m\theta}\,r\rd r\rd\theta.
\end{equation}
Therefore, the residual coefficient is the affine expression
\begin{equation}\label{eq: disc residual coefficient affine}
    a_{m,\ell}
    =\sum_{(j,k)\in\mathcal{G}_R}w_{j,k}P_{m,\ell}^{j,k}-C_{m,\ell}.
\end{equation}

For computation, choose an integer $M$ and truncate to
\begin{equation}\label{eq: disc FB truncation set}
    \mathcal{I}_M
    =\left\{(m,\ell):
    \beta_{|m|,\ell}\le \frac{\pi}{2}M\right\}.
\end{equation}
For frequency radius $R$ and coefficient budget $\lambda$, solve
\begin{equation}\label{eq: disc FB cvx program}
\begin{aligned}
    \epsilon_M=
    \min_{\{w_{j,k}\}_{(j,k)\in\mathcal{G}_R}}
    &\quad
    \sum_{(m,\ell)\in\mathcal{I}_M}
    \left|\sum_{(j,k)\in\mathcal{G}_R}w_{j,k}P_{m,\ell}^{j,k}-C_{m,\ell}\right| \\
    \text{s.t.}
    &\quad
    \sum_{(j,k)\in\mathcal{G}_R}|w_{j,k}|\le \lambda .
\end{aligned}
\end{equation}
This is again a second-order cone program. After obtaining the optimizer, one can recompute the coefficient sum on a larger set $\mathcal{I}_{M_{\mathrm{test}}}$ and use the enlarged sum as the reported disc error. This verifier, however, is time-consuming even for moderately large $M_{\mathrm{test}}$, because we do not have a simple FFT analogue for computing the Neumann Fourier--Bessel coefficients. In practice, we use a more efficient verifier, discussed in \Cref{appsubsec:disc-practical-extension-bound}. The idea is similar to the derivative bounds \eqref{eq: H2 estimate} and \eqref{eq: unit W bound e_disc}, but the practical verifier gives a much tighter upper bound through more refined calculations.

For functions that become large on $[-1,1]^2\setminus D_1$, the disc optimizer is preferable to the square optimizer in \Cref{sec: square opt} whenever we know $\norm{A}\le1$. We illustrate this with $f(z)=z^4$. This example is only illustrative, since matrix monomials can be computed by more direct quantum methods. We fix $\lambda=12$ and $M_{\text{opt}}=75$ for both the square and disc optimizers, and compare the truncation radius $R$ required for different values of $\epsilon$ in \Cref{fig:z4_cost_comparison}. The disc optimizer has a significantly lower cost, as expected.

We also remark that the Neumann Fourier--Bessel coefficient $\ell^1$ norm is usually not an optimal estimate of the local Wiener norm. Therefore, for functions that remain moderate on $[-1,1]^2\setminus D_1$, the square optimizer may give comparable or even better results.

\begin{figure}[htbp]
    \centering
    \includegraphics[width = 0.6\textwidth]{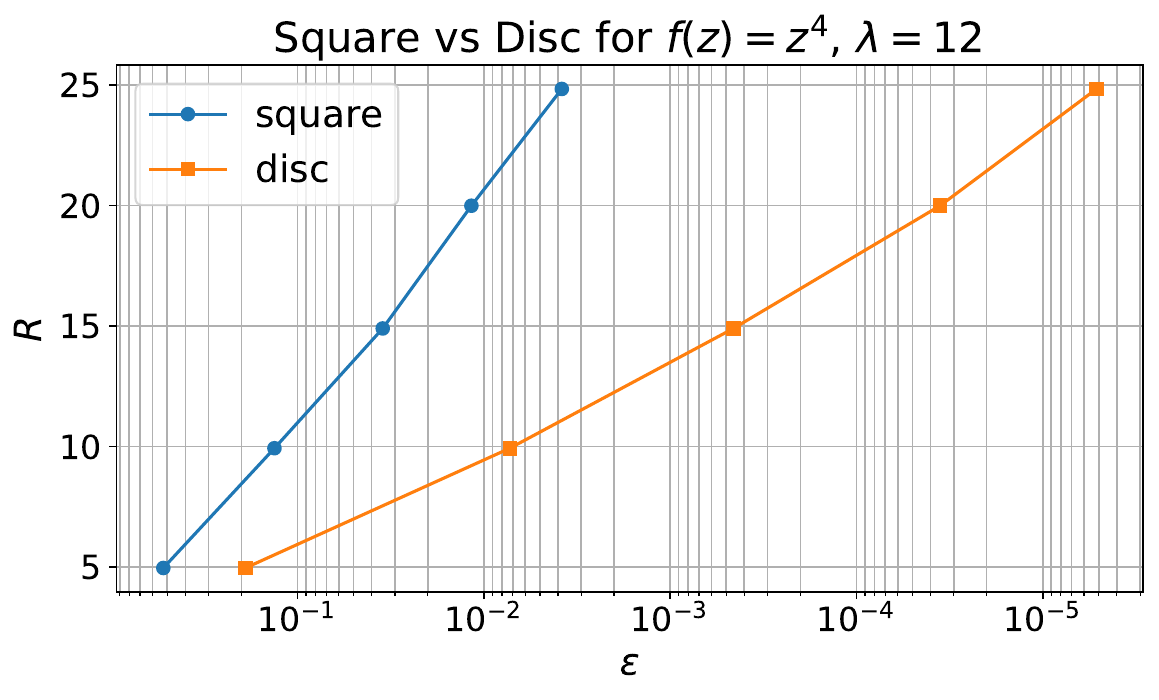}
    \caption{Required truncation radius $R$ for $f(z)=z^4$ with fixed normalization factor $\lambda = 12$. Since $\lambda$ is fixed, $R$ is proportional to the LCHS cost $\Lambda=\lambda R$.}
    \label{fig:z4_cost_comparison}
\end{figure}

\section{Quantum implementation of LCHS}\label{sec:implementation}
Given an LCHS formula of the form \eqref{eq: lchs formula} or \eqref{eq: disc lchs formula}, we briefly discuss how to implement it on a quantum computer. The goal of this section is to explain how the quantum cost depends on the LCHS parameters. The material is standard in LCHS-based methods; see \cite{an2023linear, an2026quantum, an2026laplace, low2025optimal} for more details.

There are two common ways of implementing LCHS formulas. The first is a coherent implementation, which combines LCU \cite{childs2012hamiltonian} with Hamiltonian simulation to block-encode the entire linear combination. The Hamiltonian simulation subroutine can be implemented using QSVT \cite{gilyen2019quantum}, Trotter formulas \cite{childs2021theory}, or randomized methods \cite{campbell2018random,wan2022randomized}. This is the natural choice when the goal is to prepare a state proportional to $f(A)\ket{\psi}$ or to use $f(A)$ as a subroutine in a larger quantum algorithm. The second is a hybrid implementation, designed for estimating observables such as $\bra{\psi} f(A)^\dagger O f(A)\ket{\psi}$. It samples terms from the LCHS formula classically and evaluates the corresponding quantum correlation functions one at a time. The coherent approach is usually better in asymptotic quantum resource scaling, while the hybrid approach is more suitable for near-term devices: it uses fewer ancilla qubits, avoids a large coherently controlled select oracle, and yields shallower circuits.

\subsection{Coherent implementation}\label{sec: coherent}
The coherent implementation is most directly stated for the discrete formula \eqref{eq: disc lchs formula}, which we discuss in this subsection. We begin with the standard block-encoding model for representing matrices on quantum computers.
\begin{definition}[Block-encoding {\cite{low2019hamiltonian}}]\label{def: block encoding}
    Let $A$ be a $2^n\times 2^n$ complex matrix. If a unitary $U_A\in \mathrm{U}(2^{n+b})$ acting on an $(n+b)$-qubit system has the block form
    \begin{equation}
        U_A=
        \begin{bmatrix}
            A/\alpha & \cdot \\
            \cdot & \cdot
        \end{bmatrix}
    \end{equation}
    for some $\alpha\ge \norm{A}$, then $U_A$ is called an $\alpha$-block-encoding of $A$. The normalization factor is $\alpha$, and the block-encoding uses $b$ ancilla qubits.
\end{definition}

The LCU method \cite{childs2012hamiltonian} is a standard technique for implementing matrix sums. Given unitaries $\{U_j\}$ and coefficients $\{w_j\}$, LCU implements a $\lambda$-block-encoding of $\sum_j w_jU_j$, where $\lambda=\sum_j |w_j|$. In detail, one uses the \textit{prepare} operators
\begin{equation}\label{eq: prep operators}
    \mathrm{PREP}_l\ket{0} = \frac{1}{\sqrt{\lambda}}\sum_{j=1}^J \sqrt{|w_j|}\ket{j},\qquad \mathrm{PREP}_r\ket{0} = \frac{1}{\sqrt{\lambda}}\sum_{j=1}^J \sqrt{|w_j|}e^{\I\Arg(w_j)}\ket{j},
\end{equation}
and the \textit{select} operator
\begin{equation}
    \mathrm{SEL} = \sum_{j=1}^J \ket{j}\bra{j}\otimes U_j.
\end{equation}
Then the product
\begin{equation}
    (\mathrm{PREP}_l^{\dagger}\otimes I)\ \mathrm{SEL}\ (\mathrm{PREP}_r\otimes I)
\end{equation}
is a $\lambda$-block-encoding of $\sum_j w_jU_j$.

Suppose $U_A$ and $U_{A^\dagger}$ are $\alpha$-block-encodings of $A$ and $A^\dagger$, respectively, and let $u,v\in\RR$. Applying LCU to $U_A$ and $U_{A^\dagger}$ gives an $(\alpha\sqrt{u^2+v^2})$-block-encoding of $uX+vY$, since
\begin{equation}
    uX+vY = \frac{u-\I v}{2}A + \frac{u+\I v}{2}A^\dagger,
\end{equation}
and $\frac{|u-\I v|}{2}+\frac{|u+\I v|}{2}=\sqrt{u^2+v^2}$. We can then apply QSVT-based Hamiltonian simulation \cite{gilyen2019quantum,low2019hamiltonian} to implement $e^{-\I(uX+vY)}$ with query complexity $\mo{\alpha\sqrt{u^2+v^2}+\log(\epsilon^{-1})}$.

Finally, applying LCU again to the discrete LCHS formula \eqref{eq: disc lchs formula}, with $U_j=e^{-\I(u_jX+v_jY)}$, gives a $\lambda$-block-encoding of $f(A)$ with $\lambda=\sum_j |w_j|$. The query complexity of implementing the select operator
\begin{equation}
    \mathrm{SEL} = \sum_{j=1}^J \ket{j}\bra{j}\otimes e^{-\I(u_jX+v_jY)}   
\end{equation}
is the maximum query complexity of implementing $e^{-\I(u_jX+v_jY)}$ over all $j$, namely $\mo{\alpha R+\log(\epsilon^{-1})}$. This is because the select operator can be implemented by a single Hamiltonian simulation circuit with the index register as a control. This procedure can be illustrated as
\begin{equation}
\begin{aligned}
    \text{controlled access to } U_A,U_{A^\dagger} &\xrightarrow{\text{controlled version of LCU}} \sum_{j=1}^J \ket{j}\bra{j}\otimes \begin{bmatrix} \frac{u_jX+v_jY}{\alpha R} & \cdot \\ \cdot & \cdot \end{bmatrix}\\
    &\xrightarrow{\text{QSVT Hamiltonian simulation}} \sum_{j=1}^J \ket{j}\bra{j}\otimes e^{-\I(u_jX+v_jY)} = \mathrm{SEL},
\end{aligned}
\end{equation}
and we refer readers to \cite{low2025optimal,an2023linear} for circuit details.

If we want to amplify the block-encoding normalization from $\lambda$ toward $\norm{f(A)}$, we can use uniform singular value amplification \cite{gilyen2019quantum}; the query complexity then increases by a factor of $\lambda/\norm{f(A)}$. Thus the final query complexity is proportional to $\frac{\lambda}{\norm{f(A)}}(\alpha R+\log(\epsilon^{-1}))$. Ignoring constants and using $R=\Omega(\log(\epsilon^{-1}))$, the query complexity in $A$ and $A^\dagger$ is proportional to
\begin{equation}
    \Lambda := R \lambda.
\end{equation}

If the goal is to prepare the normalized state proportional to $f(A)\ket{\psi}$ for a given $\ket{\psi}$, amplitude amplification \cite{brassard2002quantum, yoder2014fixed} also gives the same proportional dependence on $\Lambda$. This justifies using $\Lambda$ as the cost function of an LCHS formula.

\subsection{Hybrid implementation}
The hybrid implementation applies when the goal is only to evaluate an observable. Suppose $O$ is Hermitian and the desired quantity is $\bra{\psi} f(A)^\dagger O f(A)\ket{\psi}$.
Using the discrete approximation $f(A)=\sum_j w_jU_j$, this observable becomes
\begin{equation}
\label{eq: hybrid observable expansion}
    \bra{\psi} f(A)^\dagger O f(A)\ket{\psi} = \sum_{j,k=1}^J \overline{w_j}w_k
    \bra{\psi}U_j^\dagger O U_k\ket{\psi}.
\end{equation}
The hybrid method samples an index pair $(j,k)$ with probability proportional to $|w_jw_k|$, and then estimates the matrix element $o_{j,k}:=\bra{\psi}U_j^\dagger O U_k\ket{\psi}$
using the Hadamard test for non-unitary matrices \cite{tong2021fast} and amplitude estimation \cite{brassard2002quantum}. Define $Z=e^{\I \Arg(\overline{w_j}w_k)}o_{j,k}$. With sufficiently many samples, the sample average converges to $\mathbb{E}Z=\lambda^{-2}\bra{\psi}f(A)^\dagger O f(A)\ket{\psi}$, which is the desired observable up to the normalization factor $\lambda^2$.

A direct bounded-variance estimate \cite[Eq. (S37)]{an2023linear} gives sample complexity proportional to $\lambda^4\norm{O}^2/\epsilon^2$, up to logarithmic factors in the failure probability, for additive error $\epsilon$. Since each sample requires Hamiltonian simulation for time $R$, the total query complexity is proportional to $R\lambda^4\norm{O}^2/\epsilon^2$. In this setting, however, the circuit depth $R$ is often the most important resource, because hybrid implementations are usually considered for early fault-tolerant devices with limited coherence time. The same procedure also applies directly to the continuous formula \eqref{eq: lchs formula}, since one can sample $(u,v)$ from a continuous distribution.

\section{Conclusion and discussion}\label{sec:discussion}

This work shows that Weyl calculus provides a natural language for LCHS beyond the matrix exponential. The main technical advantage is that the noncommutativity of the Hermitian and skew-Hermitian parts of a matrix is absorbed into the Weyl quantization map, while the approximation problem remains a scalar approximation problem on the numerical range. In this sense, \Cref{thm: f - h bound} separates two issues that were intertwined in earlier LCHS constructions: finding a good scalar approximation to $f(x+\I y)$ and implementing the corresponding plane-wave expansion by Hamiltonian simulation. The framework also treats discrete formulas as first-class objects through the Fourier--Stieltjes viewpoint, which is important for coherent quantum implementation because the final algorithm ultimately uses a finite LCU expansion.

A first direction for future work is to make the numerical optimization framework more efficient. The convex program \eqref{eq: general opt problem} directly optimizes the discrete weights and therefore avoids both quadrature error and the choice of a parametrized ansatz. Its practical performance, however, depends on how tightly and efficiently one can upper bound the local Wiener norm of the residual. The cosine and Fourier--Bessel coefficient bounds used in this paper are convenient and convex, but they are not expected to be optimal. It would be useful to develop sharper verifiers for common domains and faster structured solvers for the resulting SOCPs.

For non-unitary dissipative dynamics simulation, our kernel improves the asymptotic prefactor of the LCHS cost to $\frac{2e}{\pi}\ln(1/\epsilon)$, which is within about a factor of $1.5$ of the lower bound from \cite{low2025optimal}. Closing this gap and matching the upper and lower bounds remains an interesting open problem.

Another question is whether Weyl-calculus-based LCHS can lead to quantum algorithms beyond eigenvalue transformation. The time-dependent dissipative linear ODE example suggests that this may be possible. It would be interesting to explore whether the framework can be adapted to other differential equations or linear algebra tasks.

\bibliographystyle{amsalpha}
\bibliography{ref}

\appendix
\section{Wiener norm and Fourier--Stieltjes algebra}
\subsection{\texorpdfstring{The Fourier--Stieltjes algebra $B(\RR^n)$}{The Fourier--Stieltjes algebra B(Rn)}}\label{appendix: B(Rn)}
Let $M(\RR^n)$ be the Banach space of finite complex Borel measures on $\RR^n$, equipped with the total variation norm. The Fourier--Stieltjes algebra $B(\RR^n)$ consists of all functions of the form
\begin{equation}\label{eq: B representation appendix}
    g(\bx) = (2 \pi)^{-n / 2} \int_{\RR^n} e^{-\I \bxi \cdot \bx} \rd \mu(\bxi),
    \qquad \mu\in M(\RR^n).
\end{equation}
The Fourier transform of a finite Borel measure is unique. Thus, if $\hat g=\mu$, we define
\begin{equation}\label{eq: B norm appendix}
    \norm{g}_W := (2 \pi)^{-n / 2} |\mu|(\RR^n).
\end{equation}
In addition to the absolutely continuous case $\mu=\hat g(\bxi)\rd\bxi$, an important example is the Dirac comb
\begin{equation*}
    \mu = \sum_{\bk\in\ZZ^n} c_{\bk} \delta_{a\bk},
\end{equation*}
whose inverse Fourier transform is the Fourier series
\begin{equation*}
    g(\bx) = (2\pi)^{-n/2}\sum_{\bk\in\ZZ^n} c_{\bk} e^{-\I a\bk \cdot \bx}.
\end{equation*}

\subsection{Local Wiener norm}\label{appendix: local wiener norm}
We now prove \Cref{eq: regular extension}, which states that the local Wiener norm on a compact set $\Omega$ can also be computed using extensions in $A(\RR^n)$ whose Fourier transforms are smooth.

Let $\Omega$ be a compact subset of $\RR^n$. For any $\epsilon > 0$, \cite[Theorem 2.6.8]{rudin1990fourier} gives a function $\chi\in A(\RR^n)$ such that $\chi=1$ on $\Omega$, $\chi$ has compact support, and $\norm{\chi}_W=(2\pi)^{-n/2}\norm{\hat{\chi}}_{L^1(\RR^n)}\le 1+\epsilon$. Since $\chi$ has compact support, $\hat{\chi}$ is smooth. For any extension $\tilde{g}\in B(\RR^n)$ of $g$, its Fourier transform $(\tilde{g})^\wedge = \tilde{\mu}$ is a finite measure. Consider the product $h = \chi \tilde{g}$. Under our Fourier convention,
\[
    \hat{h} = (2\pi)^{-n/2}\hat{\chi}*\tilde{\mu}\in C^{\infty}(\RR^n),
\]
because $\hat{\chi}$ is smooth and $\tilde{\mu}$ is finite. Moreover, Young's inequality gives
$$\norm{\hat{h}}_{L^1(\RR^n)} \le (2\pi)^{-n/2}\norm{\hat{\chi}}_{L^1(\RR^n)}\abs{\tilde{\mu}}(\RR^n) \le (1+\epsilon) \abs{\tilde{\mu}}(\RR^n).$$ 
Thus $h\in A(\RR^n)$ and $\norm{h}_{W}\le (1+\epsilon)\norm{\tilde{g}}_{W}$. Letting $\epsilon\to 0$ gives
\begin{equation}\label{eq: wiener norm equiv}
    \norm{g}_{W(\Omega)} :=\ \inf_{\tilde{g}\in B(\RR^n),\ \tilde{g}|_\Omega = g}\ \norm{\tilde{g}}_W =\ \inf_{\hat{h}\in C^{\infty}(\RR^n),\ h\in A(\RR^n),\ h|_\Omega = g}\ \norm{h}_W,
\end{equation}
completing the proof of \Cref{eq: regular extension}.

With not much effort, we can actually show the first infimum in \eqref{eq: wiener norm equiv} is attained:

\begin{lemma}\label{lem: attaining local W appendix}
Let $U=\overline{\Omega}$, where $\Omega\subset\RR^n$ is open.  Suppose that $g$ is continuous on $U$ and has at least one extension in $B(\RR^n)$.  Then there exists $G_0\in B(\RR^n)$ such that
\begin{equation*}
    G_0=g\quad\text{on }U,
    \qquad
    \norm{G_0}_W=\norm{g}_{W(U)}.
\end{equation*}
\end{lemma}
\begin{proof}
Let $m=\norm{g}_{W(U)}$. Choose $G_k\in B(\RR^n)$ with $G_k=g$ on $U$ and $\norm{G_k}_W\le m+k^{-1}$. Write $\hat G_k=\mu_k$. Then
\begin{equation*}
    |\mu_k|(\RR^n)\le (2\pi)^{n/2}(m+k^{-1})
\end{equation*}
for all $k$. Since $M(\RR^n)=C_0(\RR^n)^*$, the Banach--Alaoglu theorem gives a subnet, still denoted by $\mu_k$, converging weak-star to some finite complex Borel measure $\mu_0$. The total variation norm is weak-star lower semicontinuous, so
\begin{equation}\label{eq: lsc variation appendix}
    |\mu_0|(\RR^n)\le \liminf_k |\mu_k|(\RR^n)\le (2\pi)^{n/2}m.
\end{equation}
Let
\begin{equation*}
    G_0(\bx)=(2\pi)^{-n/2}\int_{\RR^n}e^{-\I\bxi\cdot\bx}\rd\mu_0(\bxi).
\end{equation*}
We claim that $G_0=g$ on $U$. It suffices first to prove equality in the sense of distributions on $\Omega$. Let $\varphi\in C_c^\infty(\Omega)$. Since $G_k=g$ on $\Omega$,
\begin{align*}
    \int_{\Omega} g(\bx)\varphi(\bx)\rd\bx
    &=\int_{\RR^n} G_k(\bx)\varphi(\bx)\rd\bx \\
    &=\int_{\RR^n}\hat\varphi(-\bxi)\rd\mu_k(\bxi).
\end{align*}
Here $\hat\varphi(-\bxi)\in\mathcal{S}(\RR^n)\subset C_0(\RR^n)$, so weak-star convergence gives
\begin{equation*}
    \int_{\Omega} g(\bx)\varphi(\bx)\rd\bx
    =\int_{\RR^n}\hat\varphi(-\bxi)\rd\mu_0(\bxi)
    =\int_{\RR^n} G_0(\bx)\varphi(\bx)\rd\bx.
\end{equation*}
Thus $G_0=g$ distributionally on $\Omega$. Since both $G_0$ and $g$ are continuous, they agree pointwise on $\Omega$, and hence also on $U=\overline{\Omega}$ by continuity.

Finally, \eqref{eq: lsc variation appendix} gives $\norm{G_0}_W\le m$, while the definition of $m$ gives $\norm{G_0}_W\ge m$. Therefore $\norm{G_0}_W=m$, and the infimum is attained.
\end{proof}

\subsection{\texorpdfstring{Proof of \Cref{lemma: err bounded by 2nd derivatives}}{Proof of the derivative bound}}\label{appendix: proof of lemma: err bounded by 2nd derivatives}
\begin{proof}
    Since $g$ is smooth on $D_a$, we can extend it to $D_{2a}$ using the $C^1$ Hestenes reflection \cite{hestenes1941extension}:
\begin{equation}\label{eq: Ga e_disc}
    G(r,\theta)=
    \begin{cases}
        g(r,\theta), & 0\le r\le a,\\
        \frac32 g(a,\theta)-\frac12 g(3a-2r,\theta), & a<r\le 2a.
    \end{cases}
\end{equation}
Next choose the $C^1$ cutoff function
\begin{equation}\label{eq: cutoff chi e_disc}
    \chi(t)=
    \begin{cases}
        1, & 0\le t\le 1,\\
        1-2(t-1)^2, & 1<t\le \frac32,\\
        2(2-t)^2, & \frac32<t\le 2,\\
        0, & t>2.
    \end{cases}
\end{equation}
Define the extension $E_ag$ by
\begin{equation}\label{eq: Ea extension}
    (E_ag)(r,\theta) = G(r,\theta)\chi(r/a).
\end{equation}
Then $E_ag$ is a $C^1$ function supported on $D_{2a}$. Since $G$ and $\chi(r/a)$ are continuously differentiable across the boundaries and have bounded piecewise second derivatives, $E_ag\in H^2(\RR^2)$. Hence the definition of the local Wiener norm and \eqref{eq: H2 estimate} give
    \begin{equation}\label{eq: bound by Eaf}
    \begin{aligned}
        \norm{g}_{W(D_a)}
        &\le \norm{E_ag}_W
        \le \frac{1}{2\sqrt{\pi}} \norm{(1-\Delta)(E_ag)}_{L^2(\RR^2)}  \\
        &\le \frac{\sqrt{\mathrm{Area}(D_{2a})}}{2\sqrt{\pi}}
        \max_{\RR^2} \abs{(1-\Delta)(E_ag)}
        = a \max_{D_{2a}} \abs{(1-\Delta)(E_ag)} .
    \end{aligned}
    \end{equation}

    We first bound the polar derivatives of $g$ in terms of the Cartesian supremum norms $M_k$. For $|r|\le a$,
    \begin{equation}\label{eq: bounds of df}
        \begin{aligned}
        \abs{g_r} &= \abs{g_x \cos\theta + g_y \sin\theta} \le \sqrt{2} M_1, \\
        \abs{g_{rr}} &= \abs{g_{xx} \cos^2\theta + 2g_{xy} \cos\theta\sin\theta + g_{yy} \sin^2\theta} \le 2 M_2, \\
        |g_{\theta\theta}| &\le |r|\abs{g_x\cos\theta+g_y\sin\theta} \\
        &\quad +r^2\abs{g_{xx}\sin^2\theta-2g_{xy}\sin\theta\cos\theta+g_{yy}\cos^2\theta} \\
        &\le a\sqrt{2}M_1+2a^2M_2 .
    \end{aligned}
    \end{equation}

    Inside $D_a$, $E_ag \equiv g$, yielding
    $\abs{(1-\Delta)(E_ag)} = \abs{(1-\Delta)g} \le M_0+2M_2$.
    In the annulus $D_{2a}\setminus D_a$, write $\chi_a(r)=\chi(r/a)$. The product rule gives
    \begin{equation}\label{eq: 1-delta Eaf bound}
        \abs{(1-\Delta)(E_ag)}
        \le \abs{\chi_a G} + \abs{(\Delta\chi_a)G}
        + 2\abs{\nabla\chi_a\cdot\nabla G} + \abs{(\Delta G)\chi_a}.
    \end{equation}
    For the radial cutoff function $\chi_a$, direct calculation from \eqref{eq: cutoff chi e_disc} gives
    \begin{equation*}
        \abs{\chi_a} \le 1,\quad \abs{(\chi_a)_\theta} = \abs{(\chi_a)_{\theta\theta}} = 0,\quad
        \abs{(\chi_a)_r}\le \frac{2}{a},\quad
        \abs{\Delta\chi_a} = \abs{(\chi_a)_{rr}+\frac{(\chi_a)_r}{r}}\le \frac{16}{3a^2}.
    \end{equation*}
    We bound the four terms in \eqref{eq: 1-delta Eaf bound}. For the first two terms,
    \begin{equation*}
        \abs{\chi_a G}\le 2M_0,
        \qquad
        \abs{(\Delta\chi_a)G}\le \frac{32}{3a^2}M_0\le \frac{32}{3}M_0.
    \end{equation*}
    The third term uses $G_r=g_r(3a-2r,\theta)$:
    \begin{equation*}
        2\abs{\nabla\chi_a\cdot\nabla G}
        = 2\abs{(\chi_a)_r G_r}
        \le \frac{4\sqrt{2}}{a}M_1
        \le 4\sqrt{2}M_1.
    \end{equation*}
    For the last term, with $s=3a-2r$, the polar Laplacian gives
    \begin{equation*}
    \begin{aligned}
        \abs{(\Delta G)\chi_a} \le \abs{\Delta G}
        &= \abs{\frac{3}{2r^2} g_{\theta\theta}(a,\theta)
        - \half\left(4g_{rr}(s,\theta)-\frac{2}{r} g_r(s,\theta)
        +\frac{1}{r^2} g_{\theta\theta}(s,\theta)\right) } \\
        &\le \frac{3}{2a^2}(a\sqrt{2}M_1+2a^2M_2)
        + 2(2M_2) + \frac{\sqrt{2}}{a} M_1
        + \frac{1}{2a^2}(a\sqrt{2}M_1+2a^2M_2) \\
        &\le \frac{3}{2}(\sqrt{2}M_1+2M_2) + 2(2M_2)
        + \sqrt{2} M_1 + \half (\sqrt{2}M_1+2M_2) \\
        &= 3\sqrt{2} M_1+ 8 M_2,
    \end{aligned}
    \end{equation*}
    where we used \eqref{eq: bounds of df}, $r\ge a$, and $a\ge 1$ in the annulus. Therefore, on $D_{2a}\setminus D_a$,
    \begin{equation*}
        \abs{(1-\Delta)(E_ag)}
        \le \frac{38}{3} M_0 + 7\sqrt{2} M_1 + 8 M_2
        \le 13 M_0+10M_1+8M_2.
    \end{equation*}
    Combining this with the bound inside $D_a$ and substituting into \eqref{eq: bound by Eaf} proves \eqref{eq: unit W bound e_disc}.

\end{proof}

\section{Details for the two-dimensional disc showcase}\label{appsec:disc-fb-details}

This appendix section contains the elementary calculations used in the disc showcase. The main text only uses the final coefficient formulas and the resulting second-order cone program.

\subsection{Neumann Fourier--Bessel modes and the zero convention}\label{appsubsec:disc-fb-modes}

We follow the DLMF convention for zeros of derivatives of Bessel functions \cite[\S10.21(i)]{DLMF}. Thus $\beta_{p,\ell}$ is the $\ell$-th non-negative zero of $J_p'$, $\beta_{0,1}=0$, and all other $\beta_{p,\ell}$ are positive. 

The Bessel differential equation \cite[Eq.~10.2.1]{DLMF} implies
\begin{equation}\label{eq: app bessel radial eigen equation}
    -\frac{1}{r}\frac{\rd}{\rd r}\left(r\frac{\rd}{\rd r}
    J_p(\beta r)\right)+\frac{p^2}{r^2}J_p(\beta r)
    =\beta^2J_p(\beta r).
\end{equation}
Therefore
\[
    \Psi_{m,\ell}(r,\theta)=J_{|m|}(\beta_{|m|,\ell}r)e^{\I m\theta}
\]
is an eigenfunction of $-\Delta$ on the unit disc. The boundary condition is Neumann because
\[
    \partial_r J_{|m|}(\beta_{|m|,\ell}r)\big|_{r=1}
    =\beta_{|m|,\ell}J_{|m|}'(\beta_{|m|,\ell})=0.
\]

The modes are orthogonal in $L^2(D_1)$. The angular part gives
\[
    \int_0^{2\pi}e^{\I m\theta}e^{-\I m'\theta}\rd\theta
    =2\pi\delta_{m,m'}.
\]
The radial part is the usual Sturm--Liouville orthogonality for \eqref{eq: app bessel radial eigen equation}. Thus, with
\[
    N_{p,\ell}=\int_0^1J_p(\beta_{p,\ell}r)^2r\rd r,
\]
we have
\begin{equation}\label{eq: app disc FB orthogonality}
    \int_0^{2\pi}\int_0^1
    \Psi_{m,\ell}(r,\theta)\overline{\Psi_{m',\ell'}(r,\theta)}\,r\rd r\rd\theta
    =2\pi N_{|m|,\ell}\delta_{m,m'}\delta_{\ell,\ell'}.
\end{equation}
For the constant mode, $N_{0,1}=1/2$. For $\beta_{p,\ell}>0$, the closed form
\begin{equation}\label{eq: app disc FB norm closed form}
    N_{p,\ell}
    =\frac{1}{2}\left(1-\frac{p^2}{\beta_{p,\ell}^2}\right)J_p(\beta_{p,\ell})^2
\end{equation}
follows from the Bessel integral identity \cite[Eq.~10.22.37]{DLMF} after setting the upper limit equal to $\beta_{p,\ell}$ and using $J_p'(\beta_{p,\ell})=0$.

\subsection{Jacobi--Anger expansion and the Wiener-type bound}\label{appsubsec:disc-fb-wiener-bound}

The special-function input is the Jacobi--Anger expansion. It follows from the generating function \cite[Eq.~10.12.1]{DLMF}; the Jacobi--Anger expansions are listed explicitly in \cite[Eqs.~10.12.2--10.12.3]{DLMF}. In the form used here,
\begin{equation}\label{eq: app JA expansion used}
    e^{-\I\rho r\cos(\theta-\varphi)}
    =\sum_{n\in\ZZ}(-\I)^{|n|}J_{|n|}(\rho r)e^{\I n\theta}e^{-\I n\varphi}.
\end{equation}
Indeed, the more common version is
\[
    e^{-\I z\cos\psi}=\sum_{n\in\ZZ}(-\I)^nJ_n(z)e^{\I n\psi},
\]
and \eqref{eq: app JA expansion used} follows by taking $\psi=\theta-\varphi$ and using $J_{-p}=(-1)^pJ_p$ for integer $p$ \cite[Eq.~10.4.1]{DLMF}.

Multiplying \eqref{eq: app JA expansion used} by $e^{\I m\varphi}$ and integrating in $\varphi$ gives
\begin{align*}
    \int_0^{2\pi}e^{-\I\rho r\cos(\theta-\varphi)}e^{\I m\varphi}\rd\varphi
    &=\sum_{n\in\ZZ}(-\I)^{|n|}J_{|n|}(\rho r)e^{\I n\theta}
      \int_0^{2\pi}e^{-\I n\varphi}e^{\I m\varphi}\rd\varphi  \\
    &=2\pi(-\I)^{|m|}J_{|m|}(\rho r)e^{\I m\theta}.
\end{align*}
Hence
\begin{equation}\label{eq: app disc FB plane wave average}
    J_{|m|}(\rho r)e^{\I m\theta}
    =\frac{\I^{|m|}}{2\pi}
    \int_0^{2\pi}
    e^{-\I\rho(x\cos\varphi+y\sin\varphi)}e^{\I m\varphi}\rd\varphi.
\end{equation}
Thus each Fourier--Bessel mode is an average of plane waves with frequencies on the circle $|\bxi|=\rho$. The factor $\I^{|m|}$ has modulus one, and the averaging measure in \eqref{eq: app disc FB plane wave average} has total variation one. Therefore, in the Fourier--Stieltjes sense explained in \Cref{appendix: B(Rn)},
\[
    \norm{\Psi_{m,\ell}}_{W(D_1)}\le 1.
\]

\subsection{The product integral and the plane-wave coefficient matrix}\label{appsubsec:disc-fb-matrix}

In this subsection, we give a closed-form formula for the coefficient of each plane wave in the Fourier--Bessel basis.

For $(j,k)\in\mathcal{G}_R$, write
\[
    \omega_{j,k}=\frac{\pi}{2}(j,k),
    \qquad q_{j,k}=|\omega_{j,k}|,
    \qquad \varphi_{j,k}=\Arg(j+\I k).
\]
For $(j,k)=(0,0)$, the value of $\varphi_{0,0}$ is arbitrary.
Then
\[
    e^{-\I\omega_{j,k}\cdot(x,y)}
    =e^{-\I q_{j,k}r\cos(\theta-\varphi_{j,k})}.
\]
We calculate the Fourier--Bessel coefficients by the inner products
\begin{equation}\label{eq: disc FB residual coeff P}
    P_{m,\ell}^{j,k}
    =\frac{1}{2\pi N_{|m|,\ell}}
    \int_0^{2\pi}\int_0^1
    e^{-\I q_{j,k}r\cos(\theta-\varphi_{j,k})} J_{|m|}(\beta_{|m|,\ell}r)e^{-\I m\theta}\,r\rd r\rd\theta.
\end{equation}
Substituting the Jacobi--Anger expansion \eqref{eq: app JA expansion used} into the coefficient definition \eqref{eq: disc FB residual coeff P}, the angular integral becomes
\[
    \frac{1}{2\pi}\int_0^{2\pi}e^{\I n\theta}e^{-\I m\theta}\rd\theta
    =\delta_{n,m},
\]
so only the term $n=m$ remains. The remaining radial integral gives
\[
    P_{m,\ell}^{j,k}
    =(-\I)^{|m|}e^{-\I m\varphi_{j,k}}
    \frac{1}{N_{|m|,\ell}}\int_0^1 J_{|m|}(q_{j,k}r) J_{|m|}(\beta_{|m|,\ell}r) r \rd r.
\]
The radial integral can be simplified using the standard Bessel product identity \cite[Eqs.~10.22.4--10.22.5]{DLMF}:
\begin{equation}\label{eq: app disc FB H closed form}
    \int_0^1J_p(qr)J_p(\beta r)r\rd r
    =\frac{qJ_{p+1}(q)J_p(\beta)-\beta J_p(q)J_{p+1}(\beta)}{q^2-\beta^2}.
\end{equation}
At $q=\beta$, the right-hand side is understood by continuity.

\subsection{A practical extension bound for the local Wiener norm on disc}\label{appsubsec:disc-practical-extension-bound}

The Fourier--Bessel coefficient certificate above is straightforward to incorporate into the SOCP. However, it is computationally expensive because of its slow polynomial convergence. For a fixed residual, a direct Sobolev-type extension bound can give a much less expensive a posteriori estimate of the local Wiener norm, and can sometimes be sharper than the Fourier--Bessel coefficient certificate. We record the version used in the numerical checks. It is independent of the Fourier--Bessel basis.

Let $g$ be a smooth function on $D_1$, and let $a>0$ be a free scaling parameter. By the same Cauchy--Schwarz argument as in \eqref{eq: H2 estimate}, but with the weight $1+a^2\abs{\bxi}^2$, every extension $G\in H^2(\RR^2)$ of $g$ satisfies
\begin{equation}\label{eq: app practical scaled Sobolev Wiener}
    \norm{g}_{W(D_1)}
    \le \norm{G}_{W(\RR^2)}
    \le \frac{1}{2a\sqrt{\pi}}
    \norm{G-a^2\Delta G}_{L^2(\RR^2)} .
\end{equation}
Indeed,
\[
    \int_{\RR^2}\frac{\rd \bxi}{(1+a^2\abs{\bxi}^2)^2}=\frac{\pi}{a^2},
\]
and Plancherel's theorem gives
\[
    \norm{(1+a^2\abs{\bxi}^2)\widehat G}_{L^2(\RR^2)}
    =\norm{G-a^2\Delta G}_{L^2(\RR^2)} .
\]
Thus the problem reduces to choosing an extension that makes the last norm in \eqref{eq: app practical scaled Sobolev Wiener} small.

We now compute the optimal exterior contribution on the unit disc. Write the boundary trace and the radial derivative of $g$ as
\begin{equation}\label{eq: app practical boundary Fourier data}
    g(1,\theta)=\sum_{n\in\ZZ}c_ne^{\I n\theta},
    \qquad
    \partial_rg(1,\theta)=\sum_{n\in\ZZ}d_ne^{\I n\theta}.
\end{equation}
The derivative in \eqref{eq: app practical boundary Fourier data} is the outward radial derivative from the disc. Any $H^2$ extension that equals $g$ on $D_1$ must have the same trace and radial derivative at $r=1$.

Let $E=\RR^2\setminus \overline{D_1}$ and minimize
\[
    J_a(u)=\int_E\abs{u-a^2\Delta u}^2\rd x
\]
among decaying exterior functions with the boundary data \eqref{eq: app practical boundary Fourier data}. The Euler--Lagrange equation is
\begin{equation}\label{eq: app practical EL equation}
    (1-a^2\Delta)^2u=0\qquad \text{in }E.
\end{equation}
Set $\kappa=1/a$. The decaying separated solutions of the modified Helmholtz equation are $K_{\abs{n}}(\kappa r)e^{\I n\theta}$, where $K_p$ is the modified Bessel function of the second kind; see the modified Bessel equation and its standard solutions in \cite[Eqs.~10.25.1--10.25.3]{DLMF}. Since \eqref{eq: app practical EL equation} contains the square of the modified Helmholtz operator, a second decaying radial solution is obtained by differentiating $K_{\abs{n}}(\kappa r)$ with respect to $\kappa$, namely $rK_{\abs{n}}'(\kappa r)$. Hence
\begin{equation}\label{eq: app practical exterior solution}
    u(r,\theta)
    =\sum_{n\in\ZZ}\left(A_nK_{\abs{n}}(\kappa r)
    +B_n rK_{\abs{n}}'(\kappa r)\right)e^{\I n\theta},
    \qquad r>1.
\end{equation}
Here and below, the prime denotes differentiation with respect to the argument of $K_{\abs{n}}$.

For brevity set
\[
    K=K_{\abs{n}}(\kappa),
    \qquad K'=K_{\abs{n}}'(\kappa),
    \qquad
    u_n(\kappa)=-\frac{K'}{K}.
\]
Using the modified Bessel equation at $\kappa$,
\[
    \kappa^2K''+\kappa K'-(\kappa^2+n^2)K=0,
\]
the boundary conditions $u(1,\theta)=g(1,\theta)$ and $\partial_ru(1,\theta)=\partial_rg(1,\theta)$ give, mode by mode,
\begin{equation}\label{eq: app practical AB system}
    A_nK+B_nK'=c_n,
    \qquad
    \kappa A_nK'+B_n\frac{\kappa^2+n^2}{\kappa}K=d_n.
\end{equation}
Solving the two-by-two system gives
\begin{equation}\label{eq: app practical Bn formula}
    B_n=-\frac{\kappa\left(d_n+\kappa u_n(\kappa)c_n\right)}
    {K\left(\kappa^2u_n(\kappa)^2-\kappa^2-n^2\right)} .
\end{equation}
The nontrivial factor in the denominator is positive:
\begin{equation}\label{eq: app practical positive denominator}
    \kappa^2u_n(\kappa)^2-\kappa^2-n^2>0.
\end{equation}
This follows from the integral identity \cite[Eq.~10.43.10]{DLMF}
\begin{equation}\label{eq: app practical K square integral}
    \int_\kappa^\infty K_{\abs{n}}(t)^2t\rd t
    =\frac{\kappa^2}{2}
    \left(K_{\abs{n}}'(\kappa)^2-
    \left(1+\frac{n^2}{\kappa^2}\right)K_{\abs{n}}(\kappa)^2\right).
\end{equation}

Let
\[
    v=u-a^2\Delta u.
\]
The first term in \eqref{eq: app practical exterior solution} is annihilated by $1-a^2\Delta$. Differentiating the modified Helmholtz equation with respect to $\kappa$ shows that
\[
    (1-a^2\Delta)\big(rK_{\abs{n}}'(\kappa r)e^{\I n\theta}\big)
    =-2aK_{\abs{n}}(\kappa r)e^{\I n\theta}.
\]
Therefore
\[
    v(r,\theta)
    =-2a\sum_{n\in\ZZ}B_nK_{\abs{n}}(\kappa r)e^{\I n\theta}.
\]
Using orthogonality in $\theta$, the change of variables $t=\kappa r$, the integral identity \eqref{eq: app practical K square integral}, and the expression \eqref{eq: app practical Bn formula}, we obtain
\begin{align}\label{eq: app practical exterior cost}
    \inf_u J_a(u)
    &=\norm{v}_{L^2(E)}^2 \notag\\
    &=2\pi a^2\sum_{n\in\ZZ}
    \frac{2\abs{u_n(1/a)c_n+a d_n}^2}
    {u_n(1/a)^2-(na)^2-1}.
\end{align}
Combining the interior contribution on $D_1$ with the optimal exterior contribution, \eqref{eq: app practical scaled Sobolev Wiener} gives the practical bound
\begin{equation}\label{eq: app practical Wiener final}
    \norm{g}_{W(D_1)}
    \le
    \frac{1}{2a\sqrt{\pi}}
    \left(
    \norm{g-a^2\Delta g}_{L^2(D_1)}^2
    +2\pi a^2\sum_{n\in\ZZ}
    \frac{2\abs{u_n(1/a)c_n+a d_n}^2}
    {u_n(1/a)^2-(na)^2-1}
    \right)^{1/2} .
\end{equation}
The parameter $a>0$ is free and can be optimized numerically. The case $a=1$ is the unscaled bound obtained from the operator $1-\Delta$.

\end{document}